\documentclass[runningheads, letterpaper]{llncs}
\usepackage{graphicx}

\usepackage{todonotes}
\usepackage{cite}
\usepackage{comment}
\usepackage[ruled, vlined , linesnumbered, noend]{algorithm2e}
\usepackage{amsmath}

\let\oldnl\nl
\newcommand{\nonl}{\renewcommand{\nl}{\let\nl\oldnl}}

\renewcommand{\emph}[1]{\textbf{\textit{#1}}}
\renewcommand{\subsection}[1]{\paragraph{\textbf{#1.}}}

\newcommand{\path}{\textsf{path}}

\begin{document}
\title{Exact Learning of Multitrees and Almost-Trees Using Path Queries}
%
%
\author{Ramtin~Afshar\inst{1}\orcidID{0000-0003-4740-1234} \and
Michael~T.~Goodrich\inst{1}\orcidID{0000-0002-8943-191X}}

\authorrunning{R. Afshar, M. T. Goodrich}
%
\institute{Department of Computer Science, University of California, Irvine, USA
\email{\{afsharr,goodrich\}@uci.edu}}
\maketitle              
\begin{abstract}
Given a directed graph, $G=(V,E)$,
a \emph{path} query, \textsf{path}$(u,v)$,
returns whether there is a directed path
from $u$ to $v$ in $G$, for $u,v\in V$.
Given only $V$,
exactly learning all the edges in $G$
using path queries
is often impossible, since
path queries cannot detect transitive edges.
In this paper, we
study the query complexity of exact learning for cases when learning
$G$ is possible using path queries.
In particular,
we provide efficient learning algorithms, as well as lower bounds,
for multitrees and almost-trees, including butterfly networks.

\end{abstract}

\section{Introduction}

The exact learning of a graph, which is also known as
\emph{graph reconstruction}, is the process of learning how a graph is connected
using a set of queries, each involving a subset of vertices of the
graph, to an all-knowing oracle. In this paper, we focus on learning
a directed acyclic graph (DAG) using path queries. In
particular, for a DAG, $G = (V,E)$, we are given the
vertex set, $V$, but the edge set, $E$, is unknown and learning it
through a set of path queries is our goal.
A \emph{path} query, \textsf{path}$(u,v)$, takes
two vertices, $u$ and $v$ in $V$, and returns whether
there is a directed path from $u$ to $v$ in $G$.

This work is motivated by applications in various disciplines of science, such as biology~\cite{triantafillou2017predicting, meinshausen2016methods, lagani2016probabilistic, tennant2021use}, computer science~\cite{DBLP:conf/nips/BelloH18a, DBLP:journals/jvcir/DiasGR13, DBLP:conf/icassp/BestaginiTT16, DBLP:journals/ieeemm/DiasGR13, DBLP:journals/tifs/DiasRG12, DBLP:journals/jal/GoldbergGPS98, ji2008generating, pfeffer2012malware}, economics~\cite{imbens2020potential, hunermund2019causal}, psychology~\cite{moffa2017using}, and sociology~\cite{heckerman2006bayesian}.
For instance, it can be useful for learning phylogenetic networks
from path queries. Phylogenetic networks capture
ancestry relationships between a group of objects of the same type.
For example, in a digital phylogenetic network, an object may be
a multimedia file (a video or an image)~\cite{DBLP:journals/jvcir/DiasGR13,
DBLP:conf/icassp/BestaginiTT16, DBLP:journals/ieeemm/DiasGR13,
DBLP:journals/tifs/DiasRG12}, a text
document\cite{marmerola2016reconstruction, DBLP:journals/access/ShenFRS18},
or a computer virus~\cite{DBLP:journals/jal/GoldbergGPS98,
pfeffer2012malware}.  In such a network, each node represents an
object, and directed edges show how an object has been manipulated
or edited from other objects~\cite{DBLP:conf/esa/AfsharGMO20}.  In
a digital phylogenetic network, objects are usually archived and
we can issue path queries between a pair of objects 
(see, e.g.,~\cite{DBLP:journals/jvcir/DiasGR13}). 

Learning a phylogenetic
network has several applications. For instance, learning a multimedia
phylogeny can be helpful in different areas such as security,
forensics, and copyright enforcement \cite{DBLP:journals/jvcir/DiasGR13}.
Afshar {\it et al.}~\cite{DBLP:conf/esa/AfsharGMO20} studied learning
phylogenetic trees (rooted trees) using path queries, where each
object is the result of a modification of a single parent. Our work
extends this scenario to applications where objects can be formed by merging
two or more objects into one, such as image components. 
In addition,
our work also has applications in biological 
scenarios that involve hybridization processes in
phylogenetic networks~\cite{barton2001role}.

Another application of our work is to learn the directed acyclic graph (DAG)
structure of a causal Bayesian network (CBN). 
It is well-known that observational data (collected from an undisturbed
system) is not sufficient for exact learning of the structure, and
therefore interventional data is often used, by forcing
some independent variables to take some specific values through
experiments.  An interventional path query requires a small
number of experiments, since, \path$(i,j)$, intervenes the only
variable correspondent to $i$. Therefore, applying our learning
methods (similar to the method by Bello and Honorio,
see~\cite{DBLP:conf/nips/BelloH18a}) can avoid an exponential number
of experiments~\cite{DBLP:conf/nips/KocaogluSB17},
and it can improve the results of Bello and
Honorio~\cite{DBLP:conf/nips/BelloH18a} for the types of DAGs that
we study. 

We measure the efficiency of our methods in terms of the number of vertices, $n = |V|$, using these two complexities:

\begin{itemize}
    \item Query complexity, $Q(n)$: This is the total number of queries that we perform. This parameter comes from the learning theory~\cite{DBLP:conf/birthday/AfshaniADDLM13, DBLP:journals/ai/ChoiK10, DBLP:conf/stoc/DobzinskiV12, DBLP:journals/combinatorica/Tardos89} and complexity theory~\cite{DBLP:journals/iandc/BernasconiDS01, DBLP:journals/jcss/Yao97}.

    \item Round complexity, $R(n)$: This is the number of rounds that we perform our queries. The queries performed in a round are in a batch and they may not depend on the answer of the queries in the same round (but they may depend on the queries issued in the previous rounds).
\end{itemize}

\subsection{Related Work}

The problem of exact learning of a graph using a set of queries has
been extensively studied~\cite{DBLP:conf/esa/AfsharGMO20,
DBLP:conf/spaa/AfsharGMO20, DBLP:journals/tcs/RongLYW21,
DBLP:conf/esa/Mathieu021, DBLP:conf/spaa/AfsharGMO21,
DBLP:conf/stacs/AfsharGMO22, DBLP:journals/tcs/RongYLW22,
DBLP:journals/corr/abs-2002-11541, DBLP:conf/stacs/AbrahamsenBRS16,
DBLP:conf/allerton/WangH19, DBLP:conf/alt/JagadishS13, hein1989optimal,
DBLP:conf/soda/KingZZ03, DBLP:journals/ipl/ReyzinS07}.  With regard
to previous work on learning directed graphs using path queries,
Wang and Honorio~\cite{DBLP:conf/allerton/WangH19} present a
sequential randomized algorithm that takes $Q(n) \in O(n \log^2{n})$
path queries in expectation to learn rooted trees of maximum degree,
$d$. Their divide and conquer approach is based on the notion of
an even-separator, an edge that divides the tree into two subtrees of
size at least $n/d$. Afshar {\it et al.}~\cite{DBLP:conf/esa/AfsharGMO20}
present a randomized parallel algorithm for the same problem using
$Q(n) \in O(n \log n)$ queries in $R(n) \in O(\log n)$ rounds with
high probability (w.h.p.)\footnote{We say that an event happens
  with high probability if it occurs with probability 
  at least $1 - \frac{1}{n^c}$, for some constant $c \geq 1$.}, 
which instead relies
on finding a near-separator, an edge that separates the tree into
two subtrees of size at least $n/ (d+2)$, through a ``noisy'' process
that requires noisy estimation of the number of descendants of a
node by sampling. Their method, however, relies on the fact the
ancestor set of a vertex in a rooted tree forms a total order. In
section~\ref{sec:near-trees}, we extend their work to learn a rooted
spanning tree for a DAG. Further, they show that learning a degree-$d$ rooted tree with $n$ nodes
requires $\Omega(nd + n \log n)$ path queries~\cite{DBLP:conf/esa/AfsharGMO20}.

Regarding the reconstruction of trees with a specific height,
Jagadish and Anindya \cite{DBLP:conf/alt/JagadishS13} present a
sequential deterministic algorithm to learn undirected fixed-degree
trees of height $h$ using $Q(n) \in O(n h \log{n})$ separator
queries, where a separator query given three vertices $a$, $b$, and
$c$, it returns ``true'' if and only if $b$ is on the path from $a$
to $c$. They also use this method as a subroutine to present an
algorithm to learn fixed-degree undirected trees of arbitrary height
using $O(n^{3/2} \log n)$ separator queries. Janardhanan and
Reyzin~\cite{DBLP:journals/corr/abs-2002-11541} study the problem
of learning an almost-tree of height $h$ (a directed rooted tree
with an additional cross-edge), and they present a randomized
sequential algorithm using $Q(n) \in O(n  \log^3 {n} + nh)$ queries.

\subsection{Our Contributions}

In Section~\ref{sec:multitrees}, we present our learning algorithms
for multitrees---a DAG with at most one directed path for any
two vertices. We begin, however, by first presenting a deterministic
result for learning directed rooted trees using path queries, giving
a sequential deterministic approach to learn fixed-degree
trees of height $h$, with $O(nh)$ queries, which
provides an improvement over results by Jagadish and
Anindya~\cite{DBLP:conf/alt/JagadishS13}. 
We then show how to use a tree-learning method to design 
an efficient
learning method for a multitree with $a$ roots using
$Q(n) \in O(a n \log n)$ queries and $R(n) \in  O(a \log n)$ 
rounds w.h.p. We finally show how to use our tree learning method to design
an algorithm with $Q(n) \in O(n^{3/2} \cdot \log^2{n})$ queries
to learn butterfly networks w.h.p.

In Section~\ref{sec:near-trees}, we introduce a separator theorem
for DAGs, which is useful in learning a spanning-tree of
a rooted DAG. Next, we present a parallel algorithm to learn
almost-trees of height $h$, using $O(n \log {n} + nh)$ path queries
in $O(\log n)$ parallel rounds w.h.p. We also provide a lower bound
of $\Omega(n \log {n} + nh)$ for the worst case query complexity
of a deterministic algorithm or an expected query complexity of a
randomized algorithm for learning fixed-degree almost-trees proving
that our algorithm is optimal. Moreover, this asymptotically-optimal
query complexity bound, improves the sequential query complexity
for this problem, since the best known results by Janardhanan and
Reyzin~\cite{DBLP:journals/corr/abs-2002-11541} achieved a query
complexity of $O(n \log^3 {n} + nh)$ in expectation.


\section{Preliminaries}
For a DAG, $G = (V, E)$, we represent the in-degree and out-degree
of vertex $v\in V$ with $d_i(v)$ and $d_o(v)$ respectively. Throughout
this paper, we assume that an input graph has maximum degree, $d$, i.e.,
for every $v \in V$, $d_i(v) + d_o(v) \leq d$. A vertex, $v$, is a
root of the DAG if $d_i(v)=0$. A DAG may have several roots, but
we call a DAG rooted if it has only one root.  Note that in a rooted
DAG with root $r$, there is at least one directed path from $r$ to
every $v \in V$.

\begin{definition}
\emph{(arborescence)} An arborescence is a rooted DAG with root $r$ that has exactly one path from $r$ to each vertex $v \in V$. It is also referred to as a spanning directed tree at root $r$ of a directed graph.
\end{definition}

We next introduce multitree, which is a family of DAGs useful in distributed computing~\cite{DBLP:conf/cse/ColomboLKG18, DBLP:journals/iandc/ItaiR88} that we study in Section~\ref{sec:multitrees}.

\begin{definition}
\emph{(multitree)} A multitree is a DAG in which the subgraph reachable from any vertex induces a tree, that is, it is a DAG with at most one directed path for any pair of vertices.
\end{definition}

We next review the definition of a butterfly network, which is a multitree used in high speed distributed computing~\cite{DBLP:conf/spaa/Ranade91, DBLP:journals/networks/ComellasFGM03, DBLP:conf/soda/GoodrichJS21} for which we provide efficient learning method in Section~\ref{sec:multitrees}.

\begin{definition}\label{def:butterfly}
\emph{(Butterfly network)} A butterfly network, also known as depth-$k$ Fast Fourier Transform (FFT) graph is a DAG recursively defined as $F^k = (V, E)$ as follows:
\begin{itemize}
    \item For $k=0$: $F^0$ is a single vertex, i.e. $V = \{v\}$ and $E = \{ \}$.
    \item Otherwise, suppose $F^{k-1}_A = (V_A, E_A)$ and $F^{k-1}_B = (V_B, E_B)$ each having $m$ sources and $m$ targets $(t_0, ..., t_{m-1}) \in V_A$ and $(t_m, ..., t_{2m-1}) \in V_B$. Let $V_C = (v_0, v_1, ..., v_{2m-1})$ be $2m$ additional vertices. We have $F^k = (V, E)$, where $V = V_A \cup V_B \cup V_C$ and $E = E_A \cup E_B \bigcup_{0 \leq i \leq m-1} (t_i, v_i) \cup (t_i,  v_{i+m}) \cup (t_{i+m}, v_i) \cup (t_{i+m}, v_{i+m})$ (See Figure~\ref{fig:butterfly} for illustration).
\end{itemize}
\end{definition}

\begin{figure}[b!]
\vspace{-18pt}
    \centering
\includegraphics[width = \textwidth]{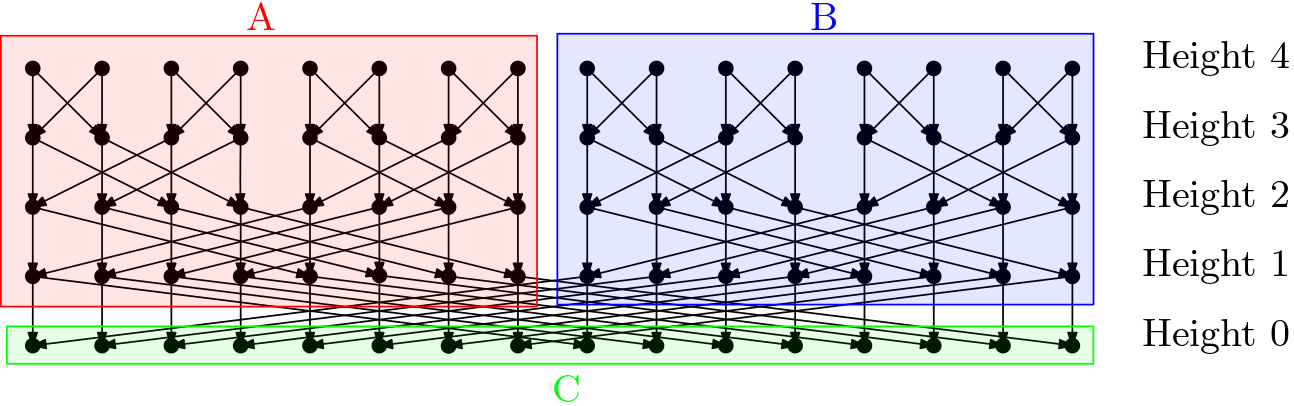}
\vspace{-24pt}
    \caption{An example of a butterfly network with height $4$ (Depth $4$), $F^4$, as a composition of two $F^3$ ($A$ and $B$) and $2^4$ additional vertices, $C$, in Height $0$.}
    \label{fig:butterfly}
   \vspace{-8pt}
\end{figure}

\begin{definition}
\emph{(ancestory)} Given a directed acyclic graph, $G = (V,E)$, we say $u$ is a \emph{parent} of a vertex $v$ ($v$ is a \emph{child} of $u$), if there exists a directed edge $(u,v) \in E$. The \emph{ancestor} relationship is a transitive closure of the parent relationship, and \emph{descendant} relationship is a transitive closure of child relationship. We denote the descendant (resp. ancestor) set of vertex $v$, with $D(v)$, (resp. $A(v)$). Also, let $C(v)$ denote children of $v$.
\end{definition}

\begin{definition}
A \emph{path query} in a directed graph, $G = (V, E)$, is a function that takes two vertices $u$ and $v$, and returns $1$, if there is a directed path from $u$ to $v$, and returns $0$ otherwise. We also let $count(s, X) =  \Sigma_{x \in X} path(s,x)$.
\end{definition}

As Wang and Honorio observed~\cite{DBLP:conf/allerton/WangH19},
transitive edges in a directed graph are not learnable by path queries. 
Thus, it is not possible using path queries 
to learn all the edges for a number of directed graph types, including
strongly connected graphs and
DAGs that are not equal to their transitive reductions (i.e., graphs that
have at least one transitive edge).
Fortunately,
transitive edges are not likely in phylogenetic networks due to their
derivative nature, so, we focus on learning DAGs without transitive edges.

\begin{definition}
In a directed graph, $G = (V, E)$, an edge $(u,v) \in E$ is called a \emph{transitive edge} if there is a directed path from $u$ to $v$ of length greater than $1$.
\end{definition}


\begin{definition}\label{def:almost-tree}
\emph{(almost-tree)} An almost-tree is a rooted DAG resulting from the union of an arborescence and an additional cross edge. 
The \emph{height} of an almost-tree is the length of its longest directed path.
\end{definition}

Note: some researchers define almost-trees to have a constant number of 
cross edges (see, e.g.,~\cite{akutsu1993polynomial,bannister2013fixed}.
But allowing more than one cross edge can cause
transitive edges; hence, almost-trees with more than one cross edge are not
all learnable using path queries, which is why we 
follow Janardhanan and Reyzin~\cite{DBLP:journals/corr/abs-2002-11541}
to limit almost-trees to have one cross edge.
We next introduce even-separator, which will be used in Section~\ref{sec:near-trees}.

\begin{definition}\label{def:even-separator}
\emph{(even-separator)} Let $G = (V, E)$ be a rooted degree-$d$ DAG. We say that vertex $v \in V$ is an even-separator if
$ \frac{|V|}{d} \leq
count(v, V)  \leq \frac{|V|  (d-1)}{d} $.
\end{definition}

\section{Learning Multitrees} \label{sec:multitrees}

In this section, we begin by presenting a deterministic algorithm to learn
a rooted tree (a multitree with a single root) of height $h$, using
$O(nh)$ path queries. Then we derive a deterministic algorithm to
learn fixed-degree directed rooted trees of any height with 
$O(n^{3/2} \sqrt{\log{n}})$ path queries. We also provide a simple reduction
from the problem of learning directed rooted trees using path queries
to the problem of learning undirected trees using separator queries,
establishing an improvement
upon results by Jagadish and Anindya~\cite{DBLP:conf/alt/JagadishS13}.
These results form building blocks for the main results of this section,
which are an efficient algorithm to learn a multitree
of arbitrary height with $a$ number of roots and 
an efficient algorithm to learn a butterfly network.

\subsection{Rooted Trees}\label{subsec:directed-trees}
Let $T = (V, E, r)$ be a directed tree rooted at $r$ with maximum degree 
that is a constant, $d$, with vertices, $V$, and edges, $E$. 
At the beginning of any exactly learning algorithm, 
we only know $V$, and $n = |V|$, and our goal is to learn $r$, and $E$ by issuing path queries. 

To begin with, learning the root of the tree can be deterministically done using $O(n)$ path queries as suggested by Afshar {\it et al.} \cite[Corollary~10]{DBLP:conf/esa/AfsharGMO20}. Their approach is to first pick an arbitrary vertex $v$, (ii) learning its ancestor set and establishing a total order on them, and (iii) finally applying a maximum-finding algorithm~\cite{DBLP:conf/stoc/ColeV86, DBLP:conf/conpar/ShiloachV81, DBLP:journals/siamcomp/Valiant75} by simulating comparisons using path queries.

Next, we show how to learn the edges, $E$. 
Jagadish and Anindya~\cite{DBLP:conf/alt/JagadishS13} proposed 
an algorithm to reconstruct
fixed-degree trees of height $h$ using $O(n h \log n)$ queries.
Their approach is to find an edge-separator---an edge that
splits the tree into two subtrees each having at least $n/d$
vertices---and then to recursively build the two subtrees.  In order
to find such an edge, (i) they pick an arbitrary vertex, $v$, and
learn an arbitrary neighbor of it such as, $u$, (ii) if $(u,v)$ is
not an edge-separator, they move to the neighboring edge that
lies on the direction of maximum vertex set size. Hence, at each
step after performing $O(n)$ queries, they get one step closer to
the edge-separator. Therefore, they learn the edge-separator using $O(nh)$ queries, and they 
incur an extra $O(\log n)$ factor to build the tree recursively
due to their edge-separator based recursive approach.

We show that finding an edge-separator for a deterministic
algorithm is unnecessary, however. We instead propose a vertex-separator
based learning algorithm.  Our \textsf{learn-short-tree} method
takes as an input, the vertex set, $V$, and root vertex, $r$, and
returns edges of the tree, $E$.  Let $\{r_1, \ldots, r_d\}$ be a
tentative set of children for vertex $r$ initially set to
$\mathit{Null}$, and for $1 \leq i \leq d$, let $V_i$ represents
the vertex set of the  subtree rooted at $r_i$. For $1 \leq i \leq
d$, we can find child $r_i$, by starting with an arbitrary vertex
$r_i$, and looping over $v \in V$ to update $r_i$ if for $v
\neq  r$, \textsf{path}$(v, r_i) = 1$. Since, in a rooted tree, an
ancestor relationship for ancestor set of any vertex is a total
order, $r_i$ will be a child of root $r$. Once we learn $r_i$, its
descendants are the set of nodes $v \in V$ such that 
\textsf{path}$(r_i, v) = 1$. We then remove $V_i$ from the set of vertices of $V$ to
learn another child of $r$ in the next iteration. It finally returns
the union of edges $(r, r_i)$ and edges returned by the recursive
calls \textsf{learn-short-tree}$(V_i, r_i)$, for $1 \leq i \leq d$
(see Algorithm~\ref{alg:learn-short-tree}).

\setlength{\textfloatsep}{10pt}
\begin{algorithm}[t]
\caption{Our algorithm to learn trees of height-$h$}
\label{alg:learn-short-tree}

\SetAlgoNoLine
\SetKwFunction{FMain}{\textsf{learn-short-tree}}
\SetKwProg{Fn}{Function}{:}{}
\Fn{\FMain{$V, r$}}{
\DontPrintSemicolon
    $E \gets \emptyset$,
    $V \gets V \setminus \{ r \}$

     \For{$i \gets 1$ to $d$}{
        $r_i \gets \mathit{Null}, V_i \gets \emptyset$
    }
     \For{$i \gets 1$ to $d$}{
        \If{$|V| \geq 1$}{
            Let $r_i$ be an arbitrary vertex in $V$

            \For{$v \in V$}{
                \lIf{\textsf{path}$(v, r_i) =1$}{
                $r_i \gets v$
                }
            }
            \For{$v \in V$}{
                \lIf{\textsf{path}$(r_i, v) =1$}{
                $V_i \gets V_i \cup \{ v \}$
                }
            }

            $V \gets V \setminus V_i$

            $E \gets E \cup (r, r_i)$

            $E \gets E \cup \textsf{small-height-tree-reconstruction}(V_i,r_i)$
        }
    }

   \KwRet $E$

}
\end{algorithm}

The query complexity, $Q(n)$, for learning the tree is as following:

\begin{equation}
Q(n) = \Sigma_{i=1}^{d} Q(|V_i|) + O(n)
\end{equation}

Since the height of the tree is reduced by at least $1$ for each recursive call, $Q(n) \in O(nh)$.  Hence, we have the following theorem.

\begin{theorem}
  One can deterministically learn a fixed-degree height-$h$ directed rooted tree using $O(nh)$ path queries.
\end{theorem}

This, in turn, implies
the following theorem (\ref{thm:arbitrary}) for 
rooted trees of arbitrary height by employing our method \textsf{learn-short-tree} in an algorithm by Jagadish and Andyia~\cite{DBLP:conf/alt/JagadishS13}, which was introduced for learning undirected trees of large height using separator queries.

\begin{theorem}\label{thm:arbitrary}
One can deterministically learn a fixed-degree directed rooted tree of arbitrary height using $O(n^{3/2} \sqrt{\log {n}})$ path queries.
\end{theorem}

\begin{proof}
 See Appendix~\ref{app:arbitrary-height}.
\end{proof}

In addition,
Theorem~\ref{thm:alg_reduction} shows that our results improve the ones by Jagadish and Anindya~\cite{DBLP:conf/alt/JagadishS13}.

\begin{theorem}\label{thm:alg_reduction}
Let $T=(V,E)$ be a fixed-degree undirected tree. If $T$ has height $h$, one can deterministically learn $T$ with $O(nh)$ separator queries, and if it has an arbitrary height, one can deterministically learn $T$ with $O(n^{3/2} \sqrt{\log {n}})$ queries.
\end{theorem}

\begin{proof}
  See Appendix~\ref{app:sep-queries}.
\end{proof}


\subsection{Multitrees of Arbitrary Height}\label{subsec:multi-arb-height}

We next provide a parallel algorithm to learn a multitree of arbitrary height with $a$ number of roots. Remind that Wang and Honorio~\cite[Theorem 8]{DBLP:conf/allerton/WangH19} prove that learning a multitree with $\Omega(n)$ roots requires $\Omega(n^2)$ queries.
Suppose that $G = (V,E)$ is a multitree with  $a$ roots. We show that we can learn $G$ using $Q(n) \in O(a n \log n)$ queries in $R(n) \in O(a \log n)$ parallel rounds w.h.p.

Let us first explain how to learn a root. Our \textsf{learn-root} method learns a root using $Q(n) \in O(n)$ queries in $R(n) \in O(1)$ rounds w.h.p. Note that in a multitree with more than one root, the ancestor set of an arbitrary vertex does not necessarily form a total order, so, we may not directly apply a parallel maximum finding algorithm on the ancestor set to learn a root.

Our \textsf{learn-root} method takes as input vertex set $V$, and
returns a root of the DAG. It first learns in parallel, $Y$, the
ancestor set of $v$ (the nodes $u \in V$ such that $path(u,v)=1$).
While $|Y|> m$, where $m = C_1 * \sqrt{|V|}$ for some constant $C_1$
fixed in the analysis, it takes a sample, $S$, of expected size of
$m$ from $Y$ uniformly at random.  Then, it performs path queries
for every pair $(a,b) \in S \times S$ in parallel to learn a partial
order of $S$, that is, we say $a < b$ if and only if $path(a,  b) = 1$. 
Hence, a root of the DAG should be an ancestor of a minimal
element in $S$. Using this fact, we keep narrowing down $Y$ until 
$ |Y| \leq m$, when we can afford to generate a partial order of $Y$
in Line~\ref{line:root-brute-force-comparison}, and return a minimal
element of $Y$ (see Algorithm~\ref{alg:learn-root}).

\begin{algorithm}[hbt]
\caption{Our algorithm to find a root in $V$}\label{alg:learn-root}
\SetKwProg{Fn}{Function}{:}{}
\SetKwFunction{FMain}{\textsf{learn-root}}
\SetAlgoNoLine
\nonl \Fn{\FMain{$V$}}{
\DontPrintSemicolon
    $m = C_1 * \sqrt{|V|}$

    Pick an arbitrary vertex $v \in V$

    \ForPar{each $u \in V$}{
        Perform query $path(u,v)$ to find ancestor set $Y$
    }

    \While{$|Y|> m $}{
        $ S \gets $ a random sample of expected size $m$ from $Y$

        \ForPar{$(a,b) \in S \times S $ }{
           \nonl Perform query $path(a,b)$
        }

            Pick a vertex $y \in S$ such that for all $a \in S$: $path(a,y)==0$

             \ForPar{$a \in Y$}{
           \nonl     Perform query $path(a,y)$ to find ancestors of $y$, $Y'$
            }
        $Y \gets Y'$
    }
    \ForPar{$(a,b) \in Y \times Y $ }{
          \nonl  Perform query $path(a,b)$ \label{line:root-brute-force-comparison}
    }

            $y \gets$  a vertex in $Y$ such that for all $a \in Y$: $path(a,y)==0$

        \KwRet{$y$}

    }

\end{algorithm}

Before providing the anlaysis of our efficient \textsf{learn-root} method, let us present Lemma~\ref{lem:elements_scattered}, which is an important lemma throughout our analysis, as it extends Afshar {\it et al.} \cite[Lemma~14]{DBLP:conf/esa/AfsharGMO20} to directed acyclic graphs.

\begin{lemma}\label{lem:elements_scattered}
Let $G=(V,E)$ be a DAG, and let $Y$ be the set of vertices formed by the union of at most $c$ directed (not necessarily disjoint) paths, where $c \leq |V|$
 and $ |Y| > m = C_1 \sqrt{|V|}$. If we take a sample, $S$, of $m$
elements from $Y$, then with probability $1 - \frac{1}{|V|^2}$, for each of these $c$ paths such as $P$, every
two consecutive nodes of $S$ in the sorted order of $P$ are within
distance $O(|Y| \log |V| / \sqrt {|V|})$ from
each other in $P$.
\end{lemma}

\begin{proof}
Since we pick our sample $S$ independently and uniformly at random, some nodes of $Y$ may be picked more than once, and each vertex will be picked with probability $p = \frac{m}{ |Y|} =  \frac{C_1 \cdot  \sqrt{ |V|}}{ |Y|} $. Let $P$ be the set of vertices of an arbitrary path among these $c$ paths. Divide $P$ into consecutive sections of size, $s = \frac{|Y| \log |V|}{ \sqrt {|V|}}$. The last section on $P$ can have any length from $1$ to $\frac{|Y| \log |V| }{ \sqrt {|V|}}$. Let $R$ be the set of vertices of an arbitrary section of path $P$ (any section except the last one). We have that expected size of $|R \cap S|$, $E[|R \cap S|]$ = $ s \cdot p = C_1 \log |V|$. Since we pick our sample independently, using standard Chernoff bound for any constant $C_1  > 8 \ln 2$, we have that $Pr[|R \cap S|=0]<1/|V|^4$. Using a union bound, with probability at least $1 - c/|V|^3$, our sample $S$ will pick at least one node from all sections except the last section of all paths. Therefore, if $c \leq |V|$, with probability at least $1 - \frac{1}{|V|^2}$, the distance between any two consecutive nodes on a path in our sample is at most $2s$.
\end{proof}

\begin{lemma}\label{lem:root_complexity}
Let $G=(V,E)$ be a DAG, and suppose that roots have at most $c \in O(n^{1/2 - \epsilon})$ for constant $0 < \epsilon < 1/2$ paths (not necessarily disjoint) in total to vertex $v$, then, \textsf{learn-root}$(V)$
outputs a root with probability at least $1-\frac{1}{|V|}$, with $Q(n) \in O(n)$ and $R(n) \in O(1)$.
\end{lemma}

\begin{proof}
The correctness of the \textsf{learn-root} method relies on the
fact that if $Y$ is a set of ancestors of vertex $v$, then for
vertex $r$, a root of the network, and for all $y \in Y$, we have:
$path(y,r)=0$. Using Lemma~\ref{lem:elements_scattered} and a union
bound, after at most $1 / \epsilon$ iterations of the \textbf{While} loop,
with probability at least $1 - \frac{1 / \epsilon}{|V|^2}$, the size of $|Y|$
will be $O(m)$. Hence, we will be able to find a root using the
queries performed in Line~\ref{line:root-brute-force-comparison}.
Note that this Las Vegas algorithm always returns a root correctly.
We can simply derive a Monte Carlo algorithm by replacing the
\textbf{while} loop with a \textbf{for} loop of two iterations.

Therefore, the query complexity of the algorithm is as follows w.h.p:
\begin{itemize}
    \item We have $O(|V|)$ queries in $1$ round to find ancestors of $v$.
    \item Then, we have $1 / \epsilon$ iterations of the \textbf{while} loop, each having $O(m^2) + O(|Y|) \in O(|V|)$ queries in $1 / \epsilon$ rounds.
    \item Finally,  we have $O(m^2)$ queries performed in $1$ round in Line~\ref{line:root-brute-force-comparison}.
\end{itemize}
Overall, this amounts to $Q(n) \in O(n)$, $R(n) \in O(1)$ w.h.p.

\end{proof}

Since in a multitree with $a \in O(n^{1/2 - \epsilon})$ roots (for $0 < \epsilon < 1/2$), each root has at most one path to a given vertex $v$, we have at most $a \in O(n^{1/2 - \epsilon})$ directed paths in total from roots to an arbitrary vertex $v$. Therefore, we can apply Lemma~\ref{lem:root_complexity} to learn a root w.h.p. Note that if $a \notin O(n^{1/2 - \epsilon})$, as an alternative, we can learn a root w.h.p. using $O(n \log n)$ queries with $R(n) \in O(\log n)$ rounds by (i) picking an arbitrary vertex $v \in V$ and learning its ancestors, $A(v) \cap V$ in parallel (ii) replacing \textsf{path} queries with \textsf{inverse-path} queries (\textsf{inverse-path}$(u,v)=1$ if and only if $v$ has a directed path to $u$), (ii) and applying the rooted tree learning method by Afshar {\it et al.}~\cite[Algorithm~2]{DBLP:conf/esa/AfsharGMO20} to learn the tree with inverse direction to $v$. Note that any of the leaves of the inverse tree rooted at $v$ is a root of the multitree.

Our  multitree learning algorithm works by repetitively learning a root, $r$, from the set of candidate roots, $R$ ($R = V$ at the beginning). Then, it learns a tree rooted at $R$ by calling the rooted tree learning method by Afshar {\it et al.}~\cite[Algorithm~2]{DBLP:conf/esa/AfsharGMO20}. Finally, it removes the set of vertices of the tree from $R$ to perform another iteration of the algorithm so long as $|R|>0$. We give the details of the algorithm below.

\begin{enumerate}
    \item Let $R$ be the set of candidate roots for the multitree initialized with $V$.
    \item Let $r \gets \textsf{learn-root}(R)$. \label{step:learn-root}
    \item Issue queries in parallel, $path(r,v)$ for all $v \in V$ to learn descendants, $D(r)$. 
    \item Learn the tree rooted at $r$ by calling \textsf{learn-rooted-tree}$(r, D(r))$.
    \item Let $R = R \setminus D(r)$, and if $|R|>0$ go to step~\ref{step:learn-root}.
\end{enumerate}

Theorem~\ref{thm:multitree} analyzes the complexity of our multitree learning algorithm.

\begin{theorem}\label{thm:multitree}
One can learn a multitree with $a$ roots using $Q(n) \in O( a n \log n)$ path queries in $R(n) \in O(a \log n)$ parallel rounds w.h.p.
\end{theorem}

\begin{proof}
The query complexity and the round complexity of our multitree learning method is dominated by the calls to the \textsf{learn-rooted-tree} by Afshar {\it et al.}~\cite[Algorithm~2]{DBLP:conf/esa/AfsharGMO20} which takes $Q(n) \in O(n \log n)$ queries  in $R(n) \in O(\log n)$ parallel rounds w.h.p. Hence, using a union bound and by adjusting the sampling constants for \textsf{learn-rooted-tree} by Afshar {\it et al.}~\cite[Algorithm~2]{DBLP:conf/esa/AfsharGMO20} we can establish the high probability bounds.
\end{proof}

\subsection{Butterfly Networks}
Next, we provide an algorithm to learn a butterfly network.
Suppose that $F^h = (V, E)$ is a butterfly network  with height $h$ 
(i.e.,  a depth-$h$ FFT graph, see definition~\ref{def:butterfly}). We show that we can learn $F^h$ using $Q(n) \in O(2^{3h/2} h^2)$ path queries with high probability. Note that in a butterfly networks of height $h$, the number of nodes will be $n = 2^h \cdot (h+1)$.
Also, note that the graph has a symmetry property, that is, all leaves are reachable from the root, and all roots are reachable from the leaves if we reverse the directions of the edges, and that each node but the leaves has exactly two children, and each node but the roots have exactly two parents, and so on.  Due to this symmetry property, we can apply \textsf{learn-short-tree} but with inverse path query (\textsf{inverse-path}$(u,v)=1$ if and only if $v$ has a directed path to $u$) to find the tree with inverse direction to a leaf.

Our algorithm first learns all the roots and all the leaves of the graph.
 We first perform a sequential search to find an arbitrary root of the network, $r$. Note that we can learn $r$ by picking an arbitrary vertex $x$ and looping over all the vertices and updating $x$ to $y$ if \textsf{path}$(y,x)=1$. After learning its descendants, $D(r)$, we make a call to our \textsf{learn-short-tree} method to build the tree rooted at $r$, which enables us to learn all the leaves, $L$. Then, we pick an arbitrary leaf, $l \in L$, and after learning its ancestors, $A(l)$, we call the \textsf{learn-short-tree} method (with inverse path query) to learn the tree with inverse direction to $l$, which enables us to learn all the roots, $R$. We then take two sample subsets, $S$, and $T$, of expected size $O(2^{h/2} h)$ from $R$, and $L$ respectively, and uniformly at random. We will show that the union of the edges of trees rooted at $r$ for all $r \in S$ and the inverse trees rooted at $l$ for all $l \in T$ includes all the edges of the network w.h.p. We give the details of our algorithm below.

\begin{enumerate}
    \item Learn a root, $r$, using a sequential search.
    \item Perform path queries to learn descendant set, $D(r)$, of $r$.
    \item Call \textsf{learn-short-tree}$(r, D(r))$ method to learn the leaves of the network, $L$.
    \item Let $l \in L$ be an arbitrary leaf in the network, then perform path queries to learn the ancestors of $l$, $A(l)$.
    \item Call \textsf{learn-short-tree}$(l, A(l))$ with inverse path query definition to learn the roots of the network, $R$.
    \item Pick a sample $S$ of size $c \cdot 2^{h/2} h$ from $R$, and a sample $T$ of size $c \cdot 2^{h/2} h$ from $L$ uniformly at random for a constant $c>0$.
    \item Perform queries to learn descendant set, $D(s)$, for every $s \in S$, and to learn ancestor set $A(t)$, for every $t \in T$.
    \item Call \textsf{learn-short-tree}$(s, D(s))$ to learn the tree rooted at $s$ for all $s \in S$.
    \item Call \textsf{learn-short-tree}$(t, A(t))$ using inverse reverse path query to learn the tree rooted at $t$ for all $t \in T$.
    \item Return the union of all the edges learned.

\end{enumerate}

\begin{theorem}
One can learn a butterfly network of height, $h$, using $Q(n) \in O(2^{3h/2} h^2)$ path queries with high probability.
\end{theorem}

\begin{proof}
The query complexity of the algorithm is dominated by $O(2^{h/2} h)$ 
times the running time of our \textsf{learn-short-tree} method,
which takes $O(2^h h)$ queries for each tree. Consider a directed edge
from vertex $x$ at height $k$ to vertex $y$ at height $k-1$ in the
network. If $k \leq h/2$, then $x$ has at least $2^{\lfloor{h/2}\rfloor}$
ancestors in the root, that is, $|A(x) \cap R| \geq 2^{\lfloor h/2 \rfloor}$. 
Since our sample, $S$, has an expected size of $2^{h/2} \cdot ch$, 
the expected size of $|S \cap A(x) \cap R| \geq ch/2$.
Using a standard Chernoff bound, the probability,  
$Pr[|S \cap A(x) \cap R|=0] \leq e^{-ch/4}$. Hence, for large enough $c$, this
probability is less than $1 / 2^{2h}$. Therefore, we will be able
to learn edge $(x, y)$ through a tree rooted at $s \in S$. Similarly,
we can show that if $k> h/2$,  then $y$ has at least 
$2^{\lfloor h/2\rfloor}$ descendants in the leaves, that is,
$|D(y) \cap L| \geq 2^{\lfloor h/2\rfloor }$. Since, our sample $T$, has an expected
size of $2^{h/2} \cdot ch$, the expected size of 
$|T \cap D(y) \cap L| \geq ch/2$. Using a standard Chernoff bound, the probability,
$Pr[|T \cap D(y) \cap L|=0] \leq e^{-ch/4}$. Hence, for large enough
$c$, this probability is less than $1 / 2^{2h}$. Therefore, we will
be able to learn edge $(x,y)$ through a tree inversely rooted at
$t \in T$ in this case. A union bound establishes the
high probability.
\end{proof}

\section{Parallel Learning of Almost-trees}\label{sec:near-trees}
 Let $G = (V, E)$ be an almost-tree of height $h$. We learn $G$ with $Q(n) \in O(n \log{n} + nh)$ path queries in $R(n) \in O(\log {n})$ rounds w.h.p. Note that we can learn the root of an almost-tree by Algorithm~\ref{alg:learn-root}, and given that the root has at most $2$ paths to any vertex, it will take $Q(n) \in O(n)$ queries and $R(n) \in O(1)$ w.h.p. by Lemma~\ref{lem:root_complexity}. We then learn a spanning rooted tree for it, and finally we learn the cross-edge. We will also prove that our algorithm is optimal by showing that any randomized algorithm needs an expected number of $\Omega(n \log{n} + nh)$.

\subsection{Learning an Arborescence in a DAG}

Our parallel algorithm learns an arborescence, a spanning directed rooted tree, of the graph with a divide and conquer approach based on our separator theorem, which is an extension of Afshar {\it et al.} \cite[Lemma 5]{DBLP:conf/esa/AfsharGMO20} for DAGs.

\begin{theorem}\label{theorem:has-separator}
 Every degree-$d$ rooted DAG, $G = (V, E)$, has an even-separator (see Definition~\ref{def:even-separator}).
\end{theorem}

\begin{proof}
We prove through a iterative process that there exists a vertex $v$ such that $ \frac{|V|}{d} \leq |D(v)| \leq \frac{|V| \cdot (d-1)}{d}$. Let $r$ be the root of the DAG. We have that $|D(r)| = |V|$. Since $r$ has at most $d$ children and each $v \in V$ is a descendent of at least one of the children of $r$, $r$ has a child $x$, such that $D(x) \geq |V|/d$. If $D(x) \leq \frac{|V| \cdot (d-1)}{d}$, $x$ is an even-separator. Otherwise, since $d_o(x) \leq d-1$, $x$ has a child, $y$, such that $|D(y)| \geq |V|/d$. If $|D(y)| \leq \frac{|V| \cdot (d-1)}{d}$, $y$ is an even-separator. Otherwise, we can repeat this iterative procedure with a child of $y$ having maximum number of descendants. Since, $|D(y)| < |D(x)|$, and a directed path in a DAG ends at vertices of out-degree $0$ (with no descendants), this iterative procedure will return an even-separator at some point.
\end{proof}

Next, we introduce Lemma~\ref{lem:exist-separator} which shows that for fixed-degree rooted DAGs, if we pick a vertex $v$ uniformly at random, there is an even separator in $A(v)$, ancestor set of $v$, with probability depending on $d$.

\begin{lemma}\label{lem:exist-separator}
    Let $G = (V, E)$ be a degree-$d$ DAG with root $r$, and let $v$ be a vertex chosen uniformly at random from $v$. Let $Y$ be the ancestor set for $v$ in $V$. Then, with probability at least $\frac{1}{d}$, there is an even-separator in $Y$.
\end{lemma}

\begin{proof}
By Theorem~\ref{theorem:has-separator}, $G$ has an even-separator, $e$. Since $|D(e)| \geq \frac{|V|}{d}$, with probability at least $\frac{1}{d}$, $v$ will be one of the descendants of $e$.

\end{proof}

Although a degree-$d$ rooted DAG has an even-separator, checking if a vertex is an even-separator requires a lot of queries for exact calculation of the number of descendants. Thus, we use a more relaxed version of the separator, which we call \emph{near-separator}, for our divide and conquer algorithm.

\begin{definition}\label{def:near-separator}
Let $G = (V, E)$ be a rooted degree-$d$ DAG. We say that vertex $v \in V$ is a near-separator if
$ \frac{|V|}{d+2} \leq
|D(v)|  \leq \frac{|V|  (d+1)}{d+2} $.

\end{definition}

Note that every even-separator is also a near-separator. We show if an even-separator exists among $A(v)$ for an arbitrary vertex $v$, then we can locate a near-separator among $A(v)$ w.h.p. Incidentally, Afshar {\it et al.}~\cite{DBLP:conf/esa/AfsharGMO20} used a similar divide and conquer approach to learn directed rooted trees, but their approach relied on the fact that there is exactly one path from root to every vertex of the tree. We will show how to meet the challenge of having multiple paths to a vertex from the root in learning an arborescence for a rooted DAG.

Our \textsf{learn-spanning-tree} method takes as input vertex set, $V$, of a DAG rooted at $r$, and returns the edges, $E$, of an arborescence of it. In particular, it enters a repeating \textbf{while} loop to learn a near-separator by (i) picking a random vertex $v \in V$, (ii) learning its ancestors, $Y = A(v) \cap V$, (iii) and checking if $Y$ has a near-separator, $w$, by calling \textsf{learn-separator} method, which we describe next.
Once \textsf{learn-separator} returns a vertex, $w$, we split $V$ into $V_1 = D(w) \cap V$ and $V_2  = V \setminus V_1$ given that $path(w,z)=1$ if and only if $z \in V_1$. If $ \frac{|V|}{d}  \leq |V_1| \leq  \frac{|V| (d-1) }{d}$, we verify $w$ is a near-separator. If $w$ is a near separator, then it calls \textsf{learn-parent} method, to learn a parent, $u$, for $w$.
Finally, it makes two recursive calls to learn a spanning tree rooted at $w$ for vertex set $V_1$, and a spanning tree rooted at $r$  with vertex set $V_2$ (see Algorithm~\ref{alg:spanning-tree} for details). Note that as \textsf{learn-parent} method is similar to the \textsf{learn-root} method, we explain it in Appendix~\ref{app:alg-details} (see Algorithm~\ref{alg:learn-parent}).

Next, we show how to adapt an algorithm to learn a near-separator for DAGs by extending the work of Afshar {\it{et al.}}\cite[Aglrorithm 3]{DBLP:conf/esa/AfsharGMO20}.
Our \textsf{learn-separator} method takes as input vertex $v$, its ancestors, $Y$, vertex set $V$ of a DAG rooted at $r$, and returns w.h.p. a near-separator among vertices of $Y$ provided that there is an even-separator in $Y$. If $|Y|$ is too large ($|Y| > |V|/K$), then it enters \textbf{Phase 1}. The goal of this phase is to remove the nodes that are unlikely to be a separator in order to pass a smaller set of candidate separator to \textbf{Phase 2}.
It chooses a random sample, $S$, of expected size $m = C_1 \sqrt{|V|}$, where $C_1>0$ is a fixed constant. It adds $\{v, r\}$ to the sample $S$. It then estimates $|D(s) \cap V|$ for each $s \in S$, using a random sample, $X_s$, of size $K = O(\log |V|)$ from $V$ by issuing path queries.
If all of the estimates, $count(s, X_s)$, are smaller than $\frac{K}{d+1}$, we return $\mathit{Null}$, as we argue that in this case the nodes in $Y$ do not have enough descendants to act as a separator.
Similarly, If all of the estimates, are greater than $\frac{K  d}{d+1}$, we return $\mathit{Null}$, as we show that in this case the nodes in $Y$ have too many descendants to act as a separator.
If one of these estimates for a vertex $s$ lies in the range of $[\frac{K}{d+1} , \frac{K d}{d+1}]$, we return it as a near-separator.
Otherwise, we filter the set of $Y$ by removing the nodes that are unlikely to be a separator through a call to \textsf{filter-separator} method, which we present next.
Then, we enter \textbf{Phase 2}, where for every $s \in Y$, we take a random sample $X_s$ of expected size of $O(log |V|)$ from $V$ to estimate $|D(s) \cap V|$. If one of these estimates for a vertex $s$ lies in the range of $[\frac{K}{d+1} , \frac{K d}{d+1}]$, we return it as a near-separator. We will show later that the output is a near-separator w.h.p (see Algorithm~\ref{alg:learn-separator}).

Next, let us explain our \textsf{filter-separator} method, whose purpose is to remove some of the vertices in $Y$ that are unlikely to be a separator to shrink the size of $Y$. We first establish a partial order on elements of $S$ by issuing path queries in parallel. Since there are at most $c =2$ directed paths from root to vertex $v$, for path $1 \leq i \leq c$, let $l_i \in S$ be the oldest node on path $i$ having $count(l_i, X_{l_i}) < \frac{K}{d+1}$ (resp. $g_i \in S$ be the youngest node on path $i$ having $count(g_i, X_{g_i})  > \frac{K d }{d+1}$).
We then perform queries to remove ancestors of $g_i$, and descendants of $l_i$ from $Y$. We will prove later that this filter reduce $|Y|$ considerably without filtering an even-separator. We will give the details of this method in Algorithm~\ref{alg:filter-separator}.

\setlength{\textfloatsep}{10pt}

\begin{algorithm}[t!]
\caption{Filter out the vertices unlikely to be a separator}
\label{alg:filter-separator}

  \SetKwFunction{FMain}{\textsf{filter-separator}}
  \SetKwProg{Fn}{Function}{:}{}
  \SetAlgoNoLine
   \nonl \Fn{\FMain{$S, Y, V$}}{
\DontPrintSemicolon

         \ForPar{each $\{a,b\} \in S$}{ \
            perform query $path(a,b)$
        }

        Let $P_1, P_2, \ldots, P_c$ be the $c$ paths from $r$ to $v$.

        For $1 \leq i \leq c:$ let $l_i \in (S \cap P_i)$ such that $count(l_i, X_{l_i}) < \frac{K}{d+1}$, and there exists no $b \in (S \cap A(l_i))$ where $count(b, X_b) < \frac{K}{d+1}$.

        For $1 \leq i \leq c:$ let $g_i \in (S \cap P_i)$ such that $count(g_i, X_{g_i}) > \frac{K \cdot d}{d+1}$, and there exists no $b \in (S \cap D(g_i))$ where $count(b, X_b) > \frac{K \cdot d}{d+1}$.

        \ForPar{ $1 \leq i \leq c$ and $v \in V$}{
            perform query $path(v, g_i)$ to find $(A(g_i) \cap V)$.

            Remove $(A(g_i) \cap V)$ from $Y$.

            perform query $path(l_i, v)$ to find $(D(l_i) \cap V)$.

            Remove $(D(l_i) \cap V)$ from $Y$.
        }

    }
    \KwRet $Y$

  \end{algorithm}

Lemma~\ref{lem:filter-separator} shows that our \textsf{filter-separator} efficiently in parallel eliminates the nodes that are unlikely to act as a separator. 

\begin{lemma}\label{lem:filter-separator}
Let $G = (V, E)$ be a DAG rooted at $r$, with at most $c$ directed (not necessarily disjoint) paths from $r$ to vertex $v$, and let $Y = A(v) \cap V$, and let $S$ be a random sample of expected size $m$ that includes $v$, and $r$ as well. The call to \textsf{filter-separator}$(S, Y, V)$ in line~\ref{line:call-filter} of our \textsf{learn-separator} method returns a set of size $O(c \cdot |Y| \log |V| / \sqrt{|V|})$, and If $Y$ has an even-separator, the returned set includes an even-separator with probability at least $1 - \frac{|S| + 1}{|V|^2}$.
\end{lemma}

\begin{proof}
See Appendix~\ref{app:arborescence}.
\end{proof}

Lemma~\ref{lem:learn-separator} establishes the fact that our \textsf{learn-separator} finds w.h.p. a near-separator among ancestors $A(v) \cap V$, if there is an even-separator in $A(v) \cap V$.

\begin{lemma}\label{lem:learn-separator}
Let $G = (V, E)$ be a DAG rooted at $r$, with at most $c$ directed (not necessarily disjoint) paths from $r$ to vertex $v$, and let $Y = A(v) \cap V$. If $Y$ has an even-separator, then the Algorithm~\ref{alg:learn-separator} returns a near-separator w.h.p.
\end{lemma}

\begin{proof}
See Appendix~\ref{app:arborescence}.
\end{proof}

\begin{lemma}\label{lem:learn-separator-complexity}
Let $G = (V, E)$ be a DAG rooted at $r$, with at most $c$ directed (not necessarily disjoint) paths from $r$ to vertex $v$. Then, our \textsf{learn-separator}$(v, Y, V, r)$ method, takes $Q(n) \in O(c |V|)$ queries in $R(n) \in O(1)$ rounds.
\end{lemma}

\begin{proof}
\begin{itemize}
    \item In \textbf{phase 1}, it takes $O(m K) \in o(|V|)$ queries in 1 round to estimate the number of descendants for sample $S$.

    \item The call to filter-separator in \textbf{phase 1} takes $m^2$ queries in one round to derive a partial order for $S$, and since there are at most $c$ paths from $r$ to $v$, it takes $O(c \cdot |V|)$ in one round to remove nodes from $Y$.

    \item In \textbf{Phase 2}, it takes $O(|Y| K) \in O(|V|)$ queries in $1$ round to estimate the number of descendants for all nodes of $Y$.
\end{itemize}
\end{proof}

\begin{theorem}\label{thm:DAG-complexity}
Suppose $G=(V,E)$ is a rooted DAG with $|V| = n$, and maximum constant degree, $d$, with at most constant, $c$ directed (not necessarily disjoint) paths from root, $r$, to each vertex. Algorithm~\ref{alg:spanning-tree} learns an arborescence of $G$ using $Q(n) \in O(n \log n)$ and $R(n) \in O(\log n)$ w.h.p.
\end{theorem}

\begin{proof}
See Appendix~\ref{app:arborescence}.
\end{proof}

\subsection{Learning a Cross-edge.}


 Next,  we will show that a cross-edge can be learnt using $O(nh)$ queries in just $2$ parallel rounds for an almost-tree of height $h$.
Our \textsf{learn-cross-edge} algorithm takes as input vertices $V$ and edges $E$ of an arborescence of a almost-tree, and returns the cross-edge from the source vertex, $s$, to the destination vertex, $t$. In this algorithm, we refer to $D(v)$ for a vertex $v$ as the set of descendants of $v$ according to $E$ (the  only edges learned by the arborescence). We will show later that there exists a vertex, $c$, whose parent is vertex, $v$, such that the cross-edge has to be from a source vertex $s \in D(c)$ to a destination vertex $t \in (D(v) \setminus D(c))$. In particular, this algorithm first learns $t$ and $c$ with $O(nh)$ queries in $1$ parallel round. Note that $t \in (D(v) \setminus D(c))$ is a node with maximum height having $path(c,t) = 1$.
Once it learns $t$ and $c$, then it learns source $s$, where $s \in D(c)$ is the node with minimum height satisfying $path(s,t)=1$, using $O(n)$ queries in $1$ round. We give the details in Algorithm~\ref{alg:learn-cross-edge}.

\begin{algorithm}[t]
\caption{lean a cross-edge for an almost tree}
\label{alg:learn-cross-edge}

  \SetKwFunction{FMain}{\textsf{learn-cross-edge}}
  \SetKwProg{Fn}{Function}{:}{}
  \SetAlgoNoLine
   \nonl \Fn{\FMain{$V, E$}}{
\DontPrintSemicolon

    \For{$v \in V$}{
        \For{$c \in C(v)$}{
            \ForPar{$t \in (D(V) \setminus D(c))$}{
                Perform query $path(c, t)$ \label{line:ct}
            }
        }
    }

    Let $c$ be the only node and let $t$ be the node with maximum height having $path(c, t) =1$

    \ForPar{$s \in D(c)$}{
        Perform query $path(s, t)$\label{line:st}
    }

    Let $s$ be the node with minimum height having $path(s,t) = 1$.

    \KwRet $(s, t)$
    }

  \end{algorithm}

The following lemma shows that Algorithm~\ref{alg:learn-cross-edge} correctly learns the cross-edge using $O(n h)$ queries in just $2$ rounds.

\begin{lemma}\label{lem:learn-cross-edge}
Given an arborescence with vertex set $V$, and edge set, $E$, of an almost-tree, Algorithm~\ref{alg:learn-cross-edge} learns the cross-edge using $O(n h)$ queries in $2$ rounds.
\end{lemma}

\begin{proof}
Suppose that the cross-edge is from a vertex $s$ to to a vertex $t$. Let $v$ be the least common ancestor of $s$ and $t$ in the arborescence, and let $c$ be a child of $v$ on the path from $v$ to $
s$. Since $t \in (D(v) \setminus D(c))$, we have that $path(c,t)=1$ in Line~\ref{line:ct}. Note that since there is only one cross-edge, there will be exactly one node such as $c$ satisfying $path(c,t)=1$. Note that in Line~\ref{line:ct} we can also learn $t$, which is the node with maximum height satisfying $path(c,t)=1$. Finally, we just do a parallel search in the descendant set of $c$ to learn $s$ in Line~\ref{line:st}.

We charge each $path(c,t)$ query in Line~\ref{line:ct} to the vertex $v$. Since each vertex has at most $d$ children the number of queries associated with vertex $v$ will be at most $O(|D(v)| \cdot d)$. Hence, using a double counting argument and the fact that each vertex is a descendant of $O(h)$ vertices, the sum of the queries performed Line~\ref{line:ct} will be,
$\Sigma_{v \in V}{O(|D(v)| \cdot d)} = O(n h)$. Finally, we need $O(n)$ queries $1$ round to learn $s$ in Line~\ref{line:st}.

\end{proof}

\begin{theorem}
  Given vertices, $V$, of an almost-tree, we can learn root, $r$, and the edges, $E$, using $Q(n) \in O(n \log n + nh)$ path queries, and $R(n) \in O(\log n)$ w.h.p.
\end{theorem}

\begin{proof}
Note that in almost-trees there are at most $c=2$ paths from root $r$ to each vertex. Therefore, by Lemma~\ref{lem:root_complexity}, we can learn root of the graph using $O(n)$ queries in $O(1)$ rounds with probability at least $1 - \frac{1}{|V|}$. Then, by Theorem~\ref{thm:DAG-complexity}, we can learn a spanning tree of the graph using $O(n \log n)$ queries in $O(\log n)$ rounds with probability at least $1 - \frac{1}{|V|}$. Finally, by Lemma~\ref{lem:learn-cross-edge} we can deterministically learn a cross-edge using $O(nh)$ queries in just $2$ rounds.
\end{proof}

\subsection{Lower bound}

The following lower bound improves the one by Janardhanan and Reyzin~\cite{DBLP:journals/corr/abs-2002-11541} and proves that our algorithm to learn almost-trees in optimal.

\begin{theorem}
  Let $G$ be a a degree-$d$ almost-tree of height $h$ with $n$ vertices. Learning $G$ takes $\Omega( n \log n + nh)$ queries. This lower bound holds for both worst case of a deterministic algorithm and for an expected cost of a randomized algorithm.
\end{theorem}

\begin{proof}
We use the same graph as the one used by Janardhanan and Reyzin~\cite{DBLP:journals/corr/abs-2002-11541}, but we improve their bound using an information-theoretic argument. Consider a caterpillar graph with height $h$, and a complete $d$-ary tree with $\Omega(n)$ leaves attached to the last level of it. If there is a cross-edge from one of the leaves of the caterpillar to one of the leaves of the  $d$-ary tree, it takes $\Omega(n h)$ queries involving a leaf of the caterpillar and a leaf of the $d$-ary tree. Suppose that a querier, Bob, knows the internal nodes of the $d$-ary, and he wants to know that for each leaf $l$ of the $d$-ary, what is the parent of $l$ in the $d$-ary tree. If there are $m$ leaves for the caterpillar, the number of possible $d$-ary trees will be at least $\frac{m!}{(d!)^{m/d}}$. Therofore, using an information-theoretic lower bound, we need $\Omega\left(\log \left(\frac{m!}{(d!)^{m/d}}\right)\right)$ bit of information to be able to learn the parent of the leaves of $d$-ary tree. Since the queries involving a leaf of the caterpillar and a leaf of the $d$-ary tree do not provide any information about how the $d$-ary tree is built, it takes $\Omega(n \log n)$ queries to learn the $d$-ary tree.
\end{proof}


\clearpage
\bibliographystyle{splncs04}

\bibliography{ref}

\begin{thebibliography}{10}
\providecommand{\url}[1]{\texttt{#1}}
\providecommand{\urlprefix}{URL }
\providecommand{\doi}[1]{https://doi.org/#1}

\bibitem{DBLP:conf/stacs/AbrahamsenBRS16}
Abrahamsen, M., Bodwin, G., Rotenberg, E., St{\"{o}}ckel, M.: Graph
  reconstruction with a betweenness oracle. In: Ollinger, N., Vollmer, H.
  (eds.) 33rd Symposium on Theoretical Aspects of Computer Science, {STACS}
  2016, February 17-20, 2016, Orl{\'{e}}ans, France. LIPIcs, vol.~47, pp.
  5:1--5:14. Schloss Dagstuhl - Leibniz-Zentrum f{\"{u}}r Informatik (2016).
  \doi{10.4230/LIPIcs.STACS.2016.5},
  \url{https://doi.org/10.4230/LIPIcs.STACS.2016.5}

\bibitem{DBLP:conf/birthday/AfshaniADDLM13}
Afshani, P., Agrawal, M., Doerr, B., Doerr, C., Larsen, K.G., Mehlhorn, K.: The
  query complexity of finding a hidden permutation. In: Brodnik, A.,
  L{\'{o}}pez{-}Ortiz, A., Raman, V., Viola, A. (eds.) Space-Efficient Data
  Structures, Streams, and Algorithms - Papers in Honor of J. Ian Munro on the
  Occasion of His 66th Birthday. Lecture Notes in Computer Science, vol.~8066,
  pp. 1--11. Springer (2013). \doi{10.1007/978-3-642-40273-9\_1},
  \url{https://doi.org/10.1007/978-3-642-40273-9\_1}

\bibitem{DBLP:conf/spaa/AfsharGMO20}
Afshar, R., Goodrich, M.T., Matias, P., Osegueda, M.C.: Reconstructing binary
  trees in parallel. In: Scheideler, C., Spear, M. (eds.) {SPAA} '20: 32nd
  {ACM} Symposium on Parallelism in Algorithms and Architectures, Virtual
  Event, USA, July 15-17, 2020. pp. 491--492. {ACM} (2020).
  \doi{10.1145/3350755.3400229}, \url{https://doi.org/10.1145/3350755.3400229}

\bibitem{DBLP:conf/esa/AfsharGMO20}
Afshar, R., Goodrich, M.T., Matias, P., Osegueda, M.C.: Reconstructing
  biological and digital phylogenetic trees in parallel. In: Grandoni, F.,
  Herman, G., Sanders, P. (eds.) 28th Annual European Symposium on Algorithms,
  {ESA} 2020, September 7-9, 2020, Pisa, Italy (Virtual Conference). LIPIcs,
  vol.~173, pp. 3:1--3:24. Schloss Dagstuhl - Leibniz-Zentrum f{\"{u}}r
  Informatik (2020). \doi{10.4230/LIPIcs.ESA.2020.3},
  \url{https://doi.org/10.4230/LIPIcs.ESA.2020.3}

\bibitem{DBLP:conf/spaa/AfsharGMO21}
Afshar, R., Goodrich, M.T., Matias, P., Osegueda, M.C.: Parallel network
  mapping algorithms. In: Agrawal, K., Azar, Y. (eds.) {SPAA} '21: 33rd {ACM}
  Symposium on Parallelism in Algorithms and Architectures, Virtual Event, USA,
  6-8 July, 2021. pp. 410--413. {ACM} (2021). \doi{10.1145/3409964.3461822},
  \url{https://doi.org/10.1145/3409964.3461822}

\bibitem{DBLP:conf/stacs/AfsharGMO22}
Afshar, R., Goodrich, M.T., Matias, P., Osegueda, M.C.: Mapping networks via
  parallel kth-hop traceroute queries. In: Berenbrink, P., Monmege, B. (eds.)
  39th International Symposium on Theoretical Aspects of Computer Science,
  {STACS} 2022, March 15-18, 2022, Marseille, France (Virtual Conference).
  LIPIcs, vol.~219, pp. 4:1--4:21. Schloss Dagstuhl - Leibniz-Zentrum f{\"{u}}r
  Informatik (2022). \doi{10.4230/LIPIcs.STACS.2022.4},
  \url{https://doi.org/10.4230/LIPIcs.STACS.2022.4}

\bibitem{akutsu1993polynomial}
Akutsu, T.: A polynomial time algorithm for finding a largest common subgraph
  of almost trees of bounded degree. IEICE Transactions on Fundamentals of
  Electronics, Communications and Computer Sciences  \textbf{76}(9),
  1488--1493 (1993)

\bibitem{bannister2013fixed}
Bannister, M.J., Eppstein, D., Simons, J.A.: Fixed parameter tractability of
  crossing minimization of almost-trees. In: International Symposium on Graph
  Drawing. pp. 340--351. Springer (2013)

\bibitem{barton2001role}
Barton, N.H.: The role of hybridization in evolution. Molecular ecology
  \textbf{10}(3),  551--568 (2001)

\bibitem{DBLP:conf/nips/BelloH18a}
Bello, K., Honorio, J.: Computationally and statistically efficient learning of
  causal bayes nets using path queries. In: Bengio, S., Wallach, H.M.,
  Larochelle, H., Grauman, K., Cesa{-}Bianchi, N., Garnett, R. (eds.) Advances
  in Neural Information Processing Systems 31: Annual Conference on Neural
  Information Processing Systems 2018, NeurIPS 2018, December 3-8, 2018,
  Montr{\'{e}}al, Canada. pp. 10954--10964 (2018),
  \url{https://proceedings.neurips.cc/paper/2018/hash/a0b45d1bb84fe1bedbb8449764c4d5d5-Abstract.html}

\bibitem{DBLP:journals/iandc/BernasconiDS01}
Bernasconi, A., Damm, C., Shparlinski, I.E.: Circuit and decision tree
  complexity of some number theoretic problems. Inf. Comput.  \textbf{168}(2),
  113--124 (2001). \doi{10.1006/inco.2000.3017},
  \url{https://doi.org/10.1006/inco.2000.3017}

\bibitem{DBLP:conf/icassp/BestaginiTT16}
Bestagini, P., Tagliasacchi, M., Tubaro, S.: Image phylogeny tree
  reconstruction based on region selection. In: 2016 {IEEE} International
  Conference on Acoustics, Speech and Signal Processing, {ICASSP} 2016,
  Shanghai, China, March 20-25, 2016. pp. 2059--2063. {IEEE} (2016).
  \doi{10.1109/ICASSP.2016.7472039},
  \url{https://doi.org/10.1109/ICASSP.2016.7472039}

\bibitem{DBLP:journals/ai/ChoiK10}
Choi, S., Kim, J.H.: Optimal query complexity bounds for finding graphs. Artif.
  Intell.  \textbf{174}(9-10),  551--569 (2010).
  \doi{10.1016/j.artint.2010.02.003},
  \url{https://doi.org/10.1016/j.artint.2010.02.003}

\bibitem{DBLP:conf/stoc/ColeV86}
Cole, R., Vishkin, U.: Deterministic coin tossing and accelerating cascades:
  micro and macro techniques for designing parallel algorithms. In: Hartmanis,
  J. (ed.) Proceedings of the 18th Annual {ACM} Symposium on Theory of
  Computing, May 28-30, 1986, Berkeley, California, {USA}. pp. 206--219. {ACM}
  (1986). \doi{10.1145/12130.12151}, \url{https://doi.org/10.1145/12130.12151}

\bibitem{DBLP:conf/cse/ColomboLKG18}
Colombo, C., Lepage, F., Kopp, R., Gnaedinger, E.: Two {SDN} multi-tree
  approaches for constrained seamless multicast. In: Pop, F., Negru, C.,
  Gonz{\'{a}}lez{-}V{\'{e}}lez, H., Rak, J. (eds.) 2018 {IEEE} International
  Conference on Computational Science and Engineering, {CSE} 2018, Bucharest,
  Romania, October 29-31, 2018. pp. 77--84. {IEEE} Computer Society (2018).
  \doi{10.1109/CSE.2018.00017}, \url{https://doi.org/10.1109/CSE.2018.00017}

\bibitem{DBLP:journals/networks/ComellasFGM03}
Comellas, F., Fiol, M.A., Gimbert, J., Mitjana, M.: The spectra of wrapped
  butterfly digraphs. Networks  \textbf{42}(1),  15--19 (2003).
  \doi{10.1002/net.10085}, \url{https://doi.org/10.1002/net.10085}

\bibitem{DBLP:journals/jvcir/DiasGR13}
Dias, Z., Goldenstein, S., Rocha, A.: Exploring heuristic and optimum branching
  algorithms for image phylogeny. J. Vis. Commun. Image Represent.
  \textbf{24}(7),  1124--1134 (2013). \doi{10.1016/j.jvcir.2013.07.011},
  \url{https://doi.org/10.1016/j.jvcir.2013.07.011}

\bibitem{DBLP:journals/ieeemm/DiasGR13}
Dias, Z., Goldenstein, S., Rocha, A.: Large-scale image phylogeny: Tracing
  image ancestral relationships. {IEEE} Multim.  \textbf{20}(3),  58--70
  (2013). \doi{10.1109/MMUL.2013.17},
  \url{https://doi.org/10.1109/MMUL.2013.17}

\bibitem{DBLP:journals/tifs/DiasRG12}
Dias, Z., Rocha, A., Goldenstein, S.: Image phylogeny by minimal spanning
  trees. {IEEE} Trans. Inf. Forensics Secur.  \textbf{7}(2),  774--788 (2012).
  \doi{10.1109/TIFS.2011.2169959},
  \url{https://doi.org/10.1109/TIFS.2011.2169959}

\bibitem{DBLP:conf/stoc/DobzinskiV12}
Dobzinski, S., Vondr{\'{a}}k, J.: From query complexity to computational
  complexity. In: Karloff, H.J., Pitassi, T. (eds.) Proceedings of the 44th
  Symposium on Theory of Computing Conference, {STOC} 2012, New York, NY, USA,
  May 19 - 22, 2012. pp. 1107--1116. {ACM} (2012).
  \doi{10.1145/2213977.2214076}, \url{https://doi.org/10.1145/2213977.2214076}

\bibitem{DBLP:journals/jal/GoldbergGPS98}
Goldberg, L.A., Goldberg, P.W., Phillips, C.A., Sorkin, G.B.: Constructing
  computer virus phylogenies. J. Algorithms  \textbf{26}(1),  188--208 (1998).
  \doi{10.1006/jagm.1997.0897}, \url{https://doi.org/10.1006/jagm.1997.0897}

\bibitem{gt-adfai-02}
Goodrich, M.T., Tamassia, R.: Algorithm Design and Applications. Wiley, New
  York, NY (2011)

\bibitem{DBLP:conf/soda/GoodrichJS21}
Goodrich, M.T., Jacob, R., Sitchinava, N.: Atomic power in forks: {A}
  super-logarithmic lower bound for implementing butterfly networks in the
  nonatomic binary fork-join model. In: Marx, D. (ed.) Proceedings of the 2021
  {ACM-SIAM} Symposium on Discrete Algorithms, {SODA} 2021, Virtual Conference,
  January 10 - 13, 2021. pp. 2141--2153. {SIAM} (2021).
  \doi{10.1137/1.9781611976465.128},
  \url{https://doi.org/10.1137/1.9781611976465.128}

\bibitem{heckerman2006bayesian}
Heckerman, D., Meek, C., Cooper, G.: A bayesian approach to causal discovery.
  In: Innovations in Machine Learning, pp. 1--28. Springer (2006)

\bibitem{hein1989optimal}
Hein, J.J.: An optimal algorithm to reconstruct trees from additive distance
  data. Bulletin of mathematical biology  \textbf{51}(5),  597--603 (1989)

\bibitem{hunermund2019causal}
H{\"u}nermund, P., Bareinboim, E.: Causal inference and data fusion in
  econometrics. arXiv preprint arXiv:1912.09104  (2019)

\bibitem{imbens2020potential}
Imbens, G.W.: Potential outcome and directed acyclic graph approaches to
  causality: Relevance for empirical practice in economics. Journal of Economic
  Literature  \textbf{58}(4),  1129--79 (2020)

\bibitem{DBLP:journals/iandc/ItaiR88}
Itai, A., Rodeh, M.: The multi-tree approach to reliability in distributed
  networks. Inf. Comput.  \textbf{79}(1),  43--59 (1988).
  \doi{10.1016/0890-5401(88)90016-8},
  \url{https://doi.org/10.1016/0890-5401(88)90016-8}

\bibitem{DBLP:conf/alt/JagadishS13}
Jagadish, M., Sen, A.: Learning a bounded-degree tree using separator queries.
  In: Jain, S., Munos, R., Stephan, F., Zeugmann, T. (eds.) Algorithmic
  Learning Theory - 24th International Conference, {ALT} 2013, Singapore,
  October 6-9, 2013. Proceedings. Lecture Notes in Computer Science, vol.~8139,
  pp. 188--202. Springer (2013). \doi{10.1007/978-3-642-40935-6\_14},
  \url{https://doi.org/10.1007/978-3-642-40935-6\_14}

\bibitem{DBLP:journals/corr/abs-2002-11541}
Janardhanan, M.V., Reyzin, L.: On learning a hidden directed graph with path
  queries. CoRR  \textbf{abs/2002.11541} (2020),
  \url{https://arxiv.org/abs/2002.11541}

\bibitem{ji2008generating}
Ji, J.H., Park, S.H., Woo, G., Cho, H.G.: Generating pylogenetic tree of
  homogeneous source code in a plagiarism detection system. International
  Journal of Control, Automation, and Systems  \textbf{6}(6),  809--817 (2008)

\bibitem{DBLP:conf/soda/KingZZ03}
King, V., Zhang, L., Zhou, Y.: On the complexity of distance-based evolutionary
  tree reconstruction. In: Proceedings of the Fourteenth Annual {ACM-SIAM}
  Symposium on Discrete Algorithms, January 12-14, 2003, Baltimore, Maryland,
  {USA}. pp. 444--453. {ACM/SIAM} (2003),
  \url{http://dl.acm.org/citation.cfm?id=644108.644179}

\bibitem{DBLP:conf/nips/KocaogluSB17}
Kocaoglu, M., Shanmugam, K., Bareinboim, E.: Experimental design for learning
  causal graphs with latent variables. In: Guyon, I., von Luxburg, U., Bengio,
  S., Wallach, H.M., Fergus, R., Vishwanathan, S.V.N., Garnett, R. (eds.)
  Advances in Neural Information Processing Systems 30: Annual Conference on
  Neural Information Processing Systems 2017, December 4-9, 2017, Long Beach,
  CA, {USA}. pp. 7018--7028 (2017),
  \url{https://proceedings.neurips.cc/paper/2017/hash/291d43c696d8c3704cdbe0a72ade5f6c-Abstract.html}

\bibitem{lagani2016probabilistic}
Lagani, V., Triantafillou, S., Ball, G., Tegn{\'e}r, J., Tsamardinos, I.:
  Probabilistic computational causal discovery for systems biology. Uncertainty
  in biology pp. 33--73 (2016)

\bibitem{marmerola2016reconstruction}
Marmerola, G.D., Oikawa, M.A., Dias, Z., Goldenstein, S., Rocha, A.: On the
  reconstruction of text phylogeny trees: evaluation and analysis of textual
  relationships. PloS one  \textbf{11}(12),  e0167822 (2016)

\bibitem{DBLP:conf/esa/Mathieu021}
Mathieu, C., Zhou, H.: A simple algorithm for graph reconstruction. In: Mutzel,
  P., Pagh, R., Herman, G. (eds.) 29th Annual European Symposium on Algorithms,
  {ESA} 2021, September 6-8, 2021, Lisbon, Portugal (Virtual Conference).
  LIPIcs, vol.~204, pp. 68:1--68:18. Schloss Dagstuhl - Leibniz-Zentrum
  f{\"{u}}r Informatik (2021). \doi{10.4230/LIPIcs.ESA.2021.68},
  \url{https://doi.org/10.4230/LIPIcs.ESA.2021.68}

\bibitem{meinshausen2016methods}
Meinshausen, N., Hauser, A., Mooij, J.M., Peters, J., Versteeg, P.,
  B{\"u}hlmann, P.: Methods for causal inference from gene perturbation
  experiments and validation. Proceedings of the National Academy of Sciences
  \textbf{113}(27),  7361--7368 (2016)

\bibitem{DBLP:books/daglib/0012859}
Mitzenmacher, M., Upfal, E.: Probability and Computing: Randomized Algorithms
  and Probabilistic Analysis. Cambridge University Press (2005).
  \doi{10.1017/CBO9780511813603},
  \url{https://doi.org/10.1017/CBO9780511813603}

\bibitem{moffa2017using}
Moffa, G., Catone, G., Kuipers, J., Kuipers, E., Freeman, D., Marwaha, S.,
  Lennox, B.R., Broome, M.R., Bebbington, P.: Using directed acyclic graphs in
  epidemiological research in psychosis: an analysis of the role of bullying in
  psychosis. Schizophrenia bulletin  \textbf{43}(6),  1273--1279 (2017)

\bibitem{pfeffer2012malware}
Pfeffer, A., Call, C., Chamberlain, J., Kellogg, L., Ouellette, J., Patten, T.,
  Zacharias, G., Lakhotia, A., Golconda, S., Bay, J., et~al.: Malware analysis
  and attribution using genetic information. In: 2012 7th International
  Conference on Malicious and Unwanted Software. pp. 39--45. IEEE (2012)

\bibitem{DBLP:conf/spaa/Ranade91}
Ranade, A.G.: Optimal speedup for backtrack search on a butterfly network. In:
  Leighton, T. (ed.) Proceedings of the 3rd Annual {ACM} Symposium on Parallel
  Algorithms and Architectures, {SPAA} '91, Hilton Head, South Carolina, USA,
  July 21-24, 1991. pp. 40--48. {ACM} (1991). \doi{10.1145/113379.113383},
  \url{https://doi.org/10.1145/113379.113383}

\bibitem{DBLP:journals/ipl/ReyzinS07}
Reyzin, L., Srivastava, N.: On the longest path algorithm for reconstructing
  trees from distance matrices. Inf. Process. Lett.  \textbf{101}(3),  98--100
  (2007). \doi{10.1016/j.ipl.2006.08.013},
  \url{https://doi.org/10.1016/j.ipl.2006.08.013}

\bibitem{DBLP:journals/tcs/RongLYW21}
Rong, G., Li, W., Yang, Y., Wang, J.: Reconstruction and verification of
  chordal graphs with a distance oracle. Theor. Comput. Sci.  \textbf{859},
  48--56 (2021). \doi{10.1016/j.tcs.2021.01.006},
  \url{https://doi.org/10.1016/j.tcs.2021.01.006}

\bibitem{DBLP:journals/tcs/RongYLW22}
Rong, G., Yang, Y., Li, W., Wang, J.: A divide-and-conquer approach for
  reconstruction of \{\emph{c}\({}_{\mbox{{\(\geq\)}5}}\)\}-free graphs via
  betweenness queries. Theor. Comput. Sci.  \textbf{917},  1--11 (2022).
  \doi{10.1016/j.tcs.2022.03.008},
  \url{https://doi.org/10.1016/j.tcs.2022.03.008}

\bibitem{DBLP:journals/access/ShenFRS18}
Shen, B., Forstall, C.W., de~Rezende~Rocha, A., Scheirer, W.J.: Practical text
  phylogeny for real-world settings. {IEEE} Access  \textbf{6},  41002--41012
  (2018). \doi{10.1109/ACCESS.2018.2856865},
  \url{https://doi.org/10.1109/ACCESS.2018.2856865}

\bibitem{DBLP:conf/conpar/ShiloachV81}
Shiloach, Y., Vishkin, U.: Finding the maximum, merging and sorting in a
  parallel computation model. In: H{\"{a}}ndler, W. (ed.) {CONPAR} 81:
  Conference on Analysing Problem Classes and Programming for Parallel
  Computing, N{\"{u}}rnberg, Germany, June 10-12, 1981, Proceedings. Lecture
  Notes in Computer Science, vol.~111, pp. 314--327. Springer (1981).
  \doi{10.1007/BFb0105127}, \url{https://doi.org/10.1007/BFb0105127}

\bibitem{DBLP:journals/combinatorica/Tardos89}
Tardos, G.: Query complexity, or why is it difficult to seperate {NP}
  \({}^{\mbox{a}}\) cap co np\({}^{\mbox{a}}\) from p\({}^{\mbox{a}}\) by
  random oracles a? Comb.  \textbf{9}(4),  385--392 (1989).
  \doi{10.1007/BF02125350}, \url{https://doi.org/10.1007/BF02125350}

\bibitem{tennant2021use}
Tennant, P.W., Murray, E.J., Arnold, K.F., Berrie, L., Fox, M.P., Gadd, S.C.,
  Harrison, W.J., Keeble, C., Ranker, L.R., Textor, J., et~al.: Use of directed
  acyclic graphs (dags) to identify confounders in applied health research:
  review and recommendations. International journal of epidemiology
  \textbf{50}(2),  620--632 (2021)

\bibitem{triantafillou2017predicting}
Triantafillou, S., Lagani, V., Heinze-Deml, C., Schmidt, A., Tegner, J.,
  Tsamardinos, I.: Predicting causal relationships from biological data:
  Applying automated causal discovery on mass cytometry data of human immune
  cells. Scientific reports  \textbf{7}(1),  1--11 (2017)

\bibitem{DBLP:journals/siamcomp/Valiant75}
Valiant, L.G.: Parallelism in comparison problems. {SIAM} J. Comput.
  \textbf{4}(3),  348--355 (1975). \doi{10.1137/0204030},
  \url{https://doi.org/10.1137/0204030}

\bibitem{DBLP:conf/allerton/WangH19}
Wang, Z., Honorio, J.: Reconstructing a bounded-degree directed tree using path
  queries. In: 57th Annual Allerton Conference on Communication, Control, and
  Computing, Allerton 2019, Monticello, IL, USA, September 24-27, 2019. pp.
  506--513. {IEEE} (2019). \doi{10.1109/ALLERTON.2019.8919924},
  \url{https://doi.org/10.1109/ALLERTON.2019.8919924}

\bibitem{DBLP:journals/jcss/Yao97}
Yao, A.C.: Decision tree complexity and betti numbers. J. Comput. Syst. Sci.
  \textbf{55}(1),  36--43 (1997). \doi{10.1006/jcss.1997.1495},
  \url{https://doi.org/10.1006/jcss.1997.1495}

\end{thebibliography}

\clearpage

\begin{appendix}

\section{Rooted Trees of Arbitrary Height}\label{app:arbitrary-height}

\begin{theorem}\textbf{(Theorem~\ref{thm:arbitrary})}
  We can deterministically learn a fixed-degree directed rooted tree using $O(n^{3/2} \sqrt{\log {n}})$ path queries.
\end{theorem}

\begin{proof}
 Jagadish and Anindya~\cite[Section~5.2]{DBLP:conf/alt/JagadishS13} provided an algorithm to learn undirected trees of arbitrary height with separator queries through the following subroutine: Given a tree $T$ and an arbitrary node set as root $r$, return a subgraph $T'$ such that for any path such as $P$, from $r$ to a leaf in $T$, $T'$ contains at least $n -  h$ vertices of $P$. Once they find $T'$, they use a an algorithm to learn trees of short height for the missing parts on each path. They control $h$ by a controlling parameter, $l$, where $h = n / l$. Besides, all the queries to find $T'$ are in the form ancestor queries which can be simply simulated by $O(1)$ path queries.
 Further, we can replace their short height tree learning algorithm with our \textsf{learn-short-tree} algorithm. Their algorithm takes $O(n l \log n)$ to learn $T'$, and $O(n h \log n) \in O(\frac{n^2}{ l} \log n )$ queries to learn the missing parts on the paths through their short depth tree learning method. We learn $T'$ using $O(n l \log n)$ path queries since all of their separator quries are in the form of $sep(r, x, y)$ where $r$ is the root, by simulating it with \textsf{path}$(x,y)$. Since our \textsf{learn-short-tree} method takes $O(n h) \in O(n^2/l)$ queries, if we set $l = \sqrt{n/ \log n}$, this amounts to a method using a total number of $Q(n) \in O(n^{3/2} \sqrt{\log n})$ queries to learn trees of arbitrary height.
\end{proof}

\section{Undirected Trees and Separator Queries}\label{app:sep-queries}

We now show how to adapt a
path-querying algorithm
to derive an algorithm
for learning an undirected fixed-degree tree
using separator queries.
This will establish improvements upon the results of Jagadish and Anindya \cite{DBLP:conf/alt/JagadishS13}.

\begin{definition}{(separator query)}
On an undirected tree $T=(V,E)$, a separator query is a function,
$sep: V \times V \times V \rightarrow \{ 0 , 1 \}$,
such that $sep(a,b,c)=1$ if removing vertex $b$ disconnects vertex $a$ from vertex $c$, and $sep(a,b,c)=0$ otherwise.
\end{definition}

Our separator querying method  (\textsf{learn-undirected-tree}) is based on a simple simulation of a
path-query algorithm (\textsf{learn-rooted-tree}), and an observation
that we can implement path queries using separator queries.
Given an undirected tree $T=(V,E^{\prime})$, we transform it into
a rooted directed tree $T=(V,E,r)$ by arbitrarily choosing a vertex, $r$,
as the root of the tree.
Then, we orient the edges in $E$ away from $r$.
Given this view, for each path query in our tree-reconstruction
algorithm, we note that $path(u,v)=1$ if and only if
$sep(r,u,v)=1$ (see Figure~\ref{fig:separator_reduction}). Finally, we report the edges returned in $\textsf{learn-rooted-tree}(V,r)$ with direction removed.

\begin{algorithm}[hbt]
\caption{Learn an undirected rooted tree with separator queries}
\label{alg:alg_reduction}

\SetAlgoNoLine
\SetKwFunction{FMain}{\textsf{learn-undirected-tree}}
\SetKwProg{Fn}{Function}{:}{}
\Fn{\FMain{$V$}}{
\DontPrintSemicolon

    pick a vertex $r$ arbitrarily from $V$ and set it as root.

    We define the path query $path(u,v)$ according to $sep(r,u,v)$: if $sep(r,u,v)=1$, then $path(u,v)=1$; otherwise, $path(u,v)=0$.

    $E \gets \textsf{learn-rooted-tree}(V,r)$

    $E' \gets $ edges of $E$ with direction removed

   \KwRet $E'$
}
\end{algorithm}

\begin{figure}[hbt]
    \centering
    \includegraphics[scale=.8]{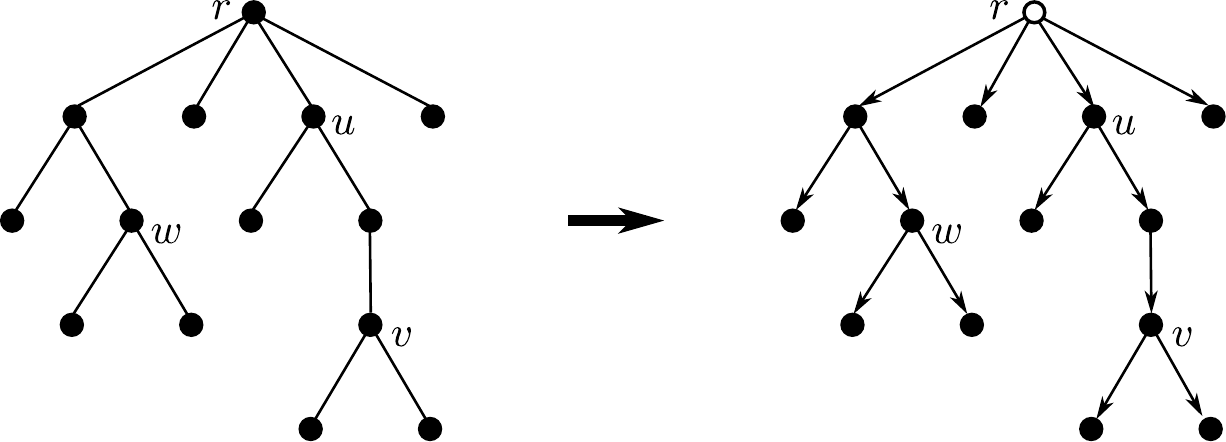}
    \caption{The reduction of separator queries (left) to path queries (right). We have that (i) $sep(r,u,v)=1 \iff path(u,v)=1$ and (ii) $sep(r,u,w)=0 \iff path(u,w)=0$.}
    \label{fig:separator_reduction}
\end{figure}

\begin{theorem}\textbf{(Theorem~\ref{thm:alg_reduction})}
Let $T=(V,E)$ be a fixed-degree undirected tree. If $T$ has height $h$, we can deterministically learn $T$ with $O(nh)$ separator queries, and if it has an arbitrary height, we can learn it with $O(n^{3/2} \sqrt{\log {n}})$ queries.
\end{theorem}

\begin{proof}
  This follows directly from our results in Subsection~\ref{subsec:directed-trees}, which establish the query query complexity of \textsf{learn-rooted-tree}, the subroutine used in Algorithm~\ref{alg:alg_reduction} that dominates the query complexity.
\end{proof}

\section{Almost-tree Algorithms in Details}\label{app:alg-details}

Algorithm~\ref{alg:spanning-tree} illustrates how our \textsf{learn-spanning-tree} method learns an arborescence for a rooted DAG in details.

\begin{algorithm}[htb!]
\caption{learn a spanning tree in a DAG}
\label{alg:spanning-tree}

\SetKwFunction{FMain}{\textsf{learn-spanning-tree}}
\SetKwProg{Fn}{Function}{:}{}
\SetKwFor{ParQuery}{for each}{query in parallel}{}
\SetAlgoNoLine
\nonl \Fn{\FMain{$V,r$}}{
\DontPrintSemicolon
  $E \gets \emptyset$

  \If(\tcp*[h]{$g$ is a chosen constant}){$|V|\le g$}{
    \KwRet edges found by a quadratic brute-force algorithm
  }
    \While{true}{
      Pick a vertex $v \in V$ uniformly at random\\
        \ForPar{$z\in V$}{ \nonl Perform query $path(z,v)$ to find $Y = A(v) \cap V$}
      w $\gets$ \textsf{learn-separator}$(v, Y ,V, r)$\\
     	\ForPar{$z\in V$}{ \nonl Perform query $path(w,z)$}
        split $V$ into $V_1,V_2$ using query results\;
      \If{w $\neq$ $\mathit{Null}$ \textbf{and} $\frac{|V|}{d+2} \leq |V_1| \leq \frac{|V| (d+1)}{d+2}$ }{
        $u \gets$ \textsf{learn-parent}$(w, V)$\\
        $E \gets E \cup \{(u,w)\}$\\

        \SetKwBlock{Pardo}{parallel do}{}
        \Pardo{
            {$E \gets  E \cup
            \textsf{learn-spanning-tree}(V_1,w)$}

            {$E \gets E \cup
            \textsf{learn-spanning-tree}(V_2,r)$}
        }
        \KwRet $E$
      }
    }
}
\end{algorithm}

Algorithm~\ref{alg:learn-separator} includes the details of how our \textsf{learn-separator} works.

\begin{algorithm}[hbt!]
\caption{For a vertex $v$, find a separator among $Y = A(v) \cap V$}
\label{alg:learn-separator}

  \SetKwFunction{FMain}{\textsf{learn-separator}}
  \SetKwProg{Fn}{Function}{:}{}
  \SetAlgoNoLine
  \nonl \Fn{\FMain{$v, Y, V, r$}}{
\DontPrintSemicolon

     \SetKwProg{phaseone}{Phase 1:}{}{}
      \SetKwProg{phasetwo}{Phase 2:}{}{}

    $m = C_1 \sqrt{|V|}$, $K = C_2 \log{|V|}$ \label{line:c1c2}

    \nonl \phaseone{}{
     \If{$|Y| > |V|/K$}{

        $S \gets$ subset of $m$ random elements from $Y$

        $S \gets S \cup{\{v,r\}}$

        \ForPar{each $s \in S$}{
            $X_s \gets$ subset of $K$ random elements from $V$

           Perform queries to find $count(s,X_s)$
        }

        \lIf{$\forall s \in S : count(s,X_s) < \frac{K}{d+1}$ }{\KwRet $ \mathit{Null}$}    \label{line:less}

          \lIf{$\forall s \in S : count(s,X_s) > \frac{K d}{d+1}$}{\KwRet $\mathit{Null}$} \label{line:more}

            \lIf{$ \exists s \in S : \frac{K}{d+1} \le count(s,X_s) \le \frac{K d}{d+1}$}{
            	\KwRet $s$ \label{line:in}
            }

       $Y \gets $ \textsf{filter-separator}$(S, Y, V)$ \label{line:call-filter}
       }
    }
    \nonl \phasetwo{}{
            \ForPar{each  $ s \in Y$}{
        $X_s \gets$ subset of $K$ random elements from $V$\\
        Perform queries to find $count(s,X_s)$
    }
    \lIf{$\exists s \in Y$\rm{ s.t. }$ \frac{K}{d+1} \leq count(s,X_s) \leq \frac{K d}{d+1}$}{
        \nonl    \KwRet $s$} \label{line:try_all}

    }
    }

    \KwRet $\mathit{Null}$ \label{line:last_Null}

  \end{algorithm}

Algorithm~\ref{alg:learn-parent}, \textsf{learn-parent} method, illustrates how we learn a parent of a given vertex $v$.  It is similar to our \textsf{learn-root} method but we explain it below for completeness. This can be used to find a parent, $p(v)$, of an arbitrary vertex $v \in V$ using $O(n)$ queries in $O(1)$ parallel rounds w.h.p. It takes as input vertex $v$, vertex set $V$, and returns a parent of $v$. It first learns in parallel, $Y$, the ancestor set of $v$ (the nodes $u \in V$ such that $path(u,v)=1$). Then, while $|Y|> m$, where $m = C_1 * \sqrt{|V|}$ for some constant $C_1$ determined in the analysis, it takes a sample $S$ of expected size of $m$ from $Y$ uniformly at random. Then, it performs path queries for every pair $(a,b) \in S \times S$ in parallel to learn a partial order of $S$, that is we say $a < b$ if and only if $path(a,  b) = 1$. Hence, a parent of $v$ should be a descendant of a maximal element in $S$. Using this fact, we continue narrowing down $Y$ until $ |Y| \leq m$, where we can afford to generate a partial order of $Y$ in Line~\ref{line:parent-brute-force-comparison}, and return a maximal element of $Y$.

\begin{algorithm}[tb]
\caption{Our algorithm to learn a parent of $v$ in $V$}\label{alg:learn-parent}
\SetKwProg{Fn}{Function}{:}{}
\SetKwFunction{FMain}{\textsf{learn-parent}}
\SetAlgoNoLine
\nonl \Fn{\FMain{$v, V$}}{
\DontPrintSemicolon
    $m = C_1 * \sqrt{|V|}$

    \ForPar{each $u \in V$}{
        Perform query $path(u,v)$ to find ancestor set $Y$
    }

    \While{$|Y|> m $}{
        $ S \gets $ a random sample of expected size $m$ from $Y$

        \ForPar{$(a,b) \in S \times S $ }{
           \nonl Perform query $path(a,b)$
        }

            Pick a vertex $y \in S$ such that for all $b \in S$: $path(y,b)==0$

            \ForPar{$b \in Y$}{
             \nonl   Perform query $path(y,b)$ to find descendants of $y$, $Y'$
            }

        $Y \gets Y'$
    }
    \ForPar{$(a,b) \in Y \times Y $ }{
          \nonl  Perform query $path(a,b)$ \label{line:parent-brute-force-comparison}
    }

            $y \gets$  a vertex in $Y$ such that for all $a \in Y$: $path(y,a)==0$

        \KwRet{$y$}

    }

\end{algorithm}

\begin{lemma}\label{lem:parent_complexity}
In a DAG, $G=(V,E)$, suppose there are at most $c \in O(n^{1/2 - \epsilon})$ for $0 < \epsilon < 1/2$ (not necessarily disjoint) directed paths from roots to vertex $v$, then \textsf{learn-parent}$(v,V)$
outputs a parent of $v$ with probability at least $1-\frac{1}{|V|}$, with $Q(n) \in O(n)$ and $R(n) \in O(1)$.
\end{lemma}

\begin{proof}
The correctness of the \textsf{learn-parent} method relies on the fact that if $Y$ is a set of ancestors of vertex $v$, then (i) for vertex $y$, a parent of $v$ (a youngest ancestor of $v$), we have that $path(y,a)==0$ for all $a \in Y$. Using Lemma~\ref{lem:elements_scattered} and a union bound, after at most $1 / \epsilon$ iterations of the \textbf{While} loop, with probability at least $1 - \frac{1 / \epsilon}{|V|^2}$, the size of $|Y|$ will be $O(m)$. Hence, we will be able to find a parent using the queries performed in Line~\ref{line:parent-brute-force-comparison}.

Therefore, the query complexity of the algorithm is as follows w.h.p:
\begin{itemize}
    \item We have $O(|V|)$ queries in $1$ round to find ancestors of $v$.
    \item Then, we have $ 1 / \epsilon$ iterations of the \textbf{while} loop, each having $O(m^2) + O(|Y|) \in O(|V|)$ queries in $1 / \epsilon$ rounds.
    \item Finally,  we have $O(m^2)$ queries performed in $1$ round in Line~\ref{line:parent-brute-force-comparison}.
\end{itemize}
Overall, this amounts to $Q(n) \in O(n)$, $R(n) \in O(1)$ w.h.p.

\end{proof}

\section{Proofs for Learning an Arborescence}\label{app:arborescence}

The following lemma shows that our \textsf{filter-separator} effectively in parallel eliminates the nodes that are unlikely to act as a separator to pass a sufficiently small set $Y$ to \textbf{Phase 2} of our \textsf{learn-separator} algorithm.

\begin{lemma}\textbf{(Lemma~\ref{lem:filter-separator})}
Let $G = (V, E)$ be a DAG rooted at $r$, with at most $c$ directed (not necessarily disjoint) paths from $r$ to vertex $v$, and let $Y = A(v) \cap V$, and let $S$ be a random sample of expected size $m$ that includes $v$, and $r$ as well. The call to \textsf{filter-separator}$(S, Y, V)$ in line~\ref{line:call-filter} of our \textsf{learn-separator} method returns a set of size $O(c \cdot |Y| \log |V| / \sqrt{|V|})$, and If $Y$ has an even-separator, the returned set includes an even-separator with probability at least $1 - \frac{|S| + 1}{|V|^2}$.
\end{lemma}

\begin{proof}
We first prove that the size of the set returned by \textsf{filter-separator} is $O(c \cdot |Y| \log |V| / \sqrt{|V|})$. We run Line~\ref{line:call-filter} of our \textsf{learn-separator} method only if we do not return in Lines~\ref{line:less}, \ref{line:more}, and \ref{line:in}; hence, for every vertex $s \in S$, we should have that $count(s, X_s) > \frac{K d}{d+1}$ or $count(s, X_s) < \frac{K}{d+1}$ and there should be nodes $\{x,y\} \subseteq S$ such that $count(x, X_x) > \frac{K d}{d+1}$ and $count(y, X_y) < \frac{K}{d+1}$.

Consider an arbitrary path, $P_i$, among these $c$ paths from $r$ to $v$. We argue that \textsf{filter-separator} returns at most $O( |Y| \log |V| / \sqrt{|V|})$ vertices of $P_i$. By Lemma~\ref{lem:elements_scattered}, with probability at least $1 - \frac{1}{|V|^2}$, the distance between any two consecutive vertices of $P_i \cap S$ is at most $O( |Y| \log |V| / \sqrt{|V|})$, and if $|P_i| > 4 |Y| \log |V| / \sqrt{|V|}$, then $|P_i \cap S| \geq 4$. Remind that $l_i \in (P_i \cap S)$ was the oldest node having $count(l_i, X_{l_i}) < \frac{K}{d+1}$ (there is no node $b \in (S \cap A(l_i))$ having $count(b, X_b) < \frac{K}{d+1}$). Also, remind that $g_i \in (P_i \cap S)$ was the youngest node having $count(g_i, X_{g_i}) > \frac{K d}{d+1}$ (there is no node $b \in (S \cap D(g_i))$ having $count(b, X_b) > \frac{K d}{d+1}$). Since we remove ancestors of $g_i$, and descendants of $l_i$ from $P_i$, our \textsf{filter-separator} algorithm returns at most vertices between two consecutive vertices of $P_i \cap S$, having a size of $O( |Y| \log |V| / \sqrt{|V|})$. Since, we have at most $c$ paths, therefore the size of the set returned by this algorithm is at most $O( c \cdot |Y| \log |V| / \sqrt{|V|})$.

Let $e$ be an even-separator in $Y$. Next, we prove that the returned set includes $e$. Afshar {\it et al.}~\cite{DBLP:conf/esa/AfsharGMO20} showed that there exists a constant $C_2 > 0$, as used in Line~\ref{line:c1c2} of our \textsf{learn-separator} algorithm such that the following probability bound hold:

\begin{equation}\label{eq:1}
  \begin{cases}

        \Pr\left(count(s,X_s) \geq \frac{K}{d+1}\right) \geq 1 - \frac{1}{|V|^2} & \mbox{if }count(s,V) \geq \frac{|V|}{d}, \\
        \Pr\left(count(s,X_s) \leq K  \frac{d}{d+1}\right) \geq 1 - \frac{1}{|V|^2} & \mbox{if } count(s,V) \leq |V|  \frac{d-1}{d},\\
               \end{cases}
 \end{equation}
 \begin{equation}\label{eq:2}
  \begin{cases}

        \Pr\left(count(s,X_s) < \frac{K}{d+1}\right) \geq 1 - \frac{1}{|V|^2} & \mbox{if }count(s,V) < \frac{|V|}{d+2}, \\
        \Pr\left(count(s,X_s) > K  \frac{d}{d+1}\right) \geq 1 - \frac{1}{ |V|^2} & \mbox{if }count(s,V) > |V|  \frac{d+1}{d+2} \\

        \end{cases}
 \end{equation}

Let $s \in S$ be an arbitrary ancestor of $e$. Hence, $count(s,V) \geq count(e, V) \geq \frac{|V|}{d}$. By Equation~\ref{eq:1}, with probability at least $1 - \frac{1}{|V|^2}$, $count(s, X_s) \geq \frac{K}{d+1}$. Hence, $s$ cannot be equal with $l_i$, for $1 \leq i \leq c$, and therefore, we do not remove descendants of $s$.
Similarly, for an arbitrary descendant, $s \in S$, of $e$, $count(s,V) \leq count(e, V) \leq \frac{|V| (d-1)}{d}$. By Equation~\ref{eq:1}, with probability at least $1 - \frac{1}{|V|^2}$, $count(s, X_s) \leq \frac{K d}{d+1}$. Hence, $s$ cannot be equal with $g_i$, for $1 \leq i \leq c$, and therefore, we do not remove ancestors of $s$. Therefore, using a union bound, with probability at least $1 - \frac{|S|}{|V|^2}$, we do not remove $e$ from $Y$.

All together, using a union bound with probability at least $1 - \frac{|S| + 1}{|V|^2}$, the call to our \textsf{filter-separator} in Line~\ref{line:call-filter} of our \textsf{learn-separator} algorithm returns a set of size $O(c \cdot |Y| \log |V| / \sqrt{|V|})$ without filtering an even-separator.
\end{proof}

The following lemma shows the efficiency of \textsf{learn-separator} method.

\begin{lemma}\textbf{(Lemma~\ref{lem:learn-separator})}
Let $G = (V, E)$ be a DAG rooted at $r$, with at most constant $c$ directed (not necessarily disjoint) paths from $r$ to vertex $v$, and let $Y = A(v) \cap V$. If $Y$ has an even-separator, then the Algorithm~\ref{alg:learn-separator} returns a near-separator w.h.p.
\end{lemma}

\begin{proof}

 We will show that the probability of returning $\mathit{Null}$ or returning a vertex that is not a near-separator is at most $\frac{1}{|V|}$. Let $e \in Y$ be an even-separator. We evaluate this probability according to different lines of the algorithm.

 \begin{itemize}
     \item Lines~\ref{line:less},~\ref{line:more}: Since $e$ is an even-separator, we have $\frac{|V|}{d} \leq count(e, V) \leq \frac{|V| (d-1)}{d}$. On the other hand, as $r \in A(v)$ we have that $count(r,V) \geq count(e, V) \geq \frac{|V|}{d}$. Also, since $v \in D(e)$, we can establish that $count(v,V) \leq count(e, V) \leq \frac{|V| (d-1)}{d}$. Hence, by Equation~\ref{eq:1} and a union bound, with probability at least $1  - \frac{2}{|V|^2}$, we have that $count(r, X_r) \geq \frac{K}{d+1}$, and $count(v, X_v) \leq \frac{K d}{d+1}$, and therefore we do not return $\mathit{Null}$ in Lines~\ref{line:less},~\ref{line:more}.

     \item Line~\ref{line:in}: Consider a vertex $s \in S$ with $count(s, V) < \frac{|V|}{d+2}$. By Equation~\ref{eq:2}, with probability at least $1 - \frac{1}{|V|^2}$, $count(s, X_s) < \frac{K}{d+1}$, and therefore, $s$ is not picked in Line~\ref{line:in} as a near-separator. Similarly, for vertex $s \in S$ with $count(s, V) > \frac{|V| (d+1)}{d+2}$, by Equation~\ref{eq:2}, with probability at least $1 - \frac{1}{|V|^2}$, $count(s, X_s) > \frac{K d}{d+1}$, indicating that we do not pick $s$ as a near-separator Line~\ref{line:in}. Hence, using a union bound with probability at least $1 - \frac{|S|}{|V|^2}$, the returned vertex in Line~\ref{line:in} is a near-separator.

     \item Line~\ref{line:try_all}: Remind that in Lemma~\ref{lem:filter-separator}, we proved that with probability at least $1 - \frac{|S| +1 }{|V|^2}$, our call to \textsf{filter-separator} method in Line~\ref{line:call-filter} of our \textsf{learn-separator} algorithm returns a set including an even-separator, $e$, and the set size of $O\left(c \cdot |Y| \frac{\log |V|}{\sqrt{|V|}}\right) \in O(\sqrt{|V|} \log{|V|})$. Using an argument similar to the one for Line~\ref{line:in}, we can show that with probability at least $1 - \frac{O(\sqrt{|V|} \log{|V|})}{|V|^2}$, if we return a vertex in Line~\ref{line:try_all}, it will be near-separator.

     \item Line~\ref{line:last_Null}: Note that the set returned by \textsf{filter-separator} method includes $e$. By Equation~\ref{eq:1}, with probability at least $1 - \frac{2}{|V|^2}$, $\frac{K}{d+1} \leq count(e, X_e) \leq \frac{K d}{d+1}$ in Line~\ref{line:try_all}, and therefore, we do not get to run Line~\ref{line:last_Null}.

 \end{itemize}
Therefore, using a union bound, we can show that if $Y$ has an even separator, with probability at least $1  - \frac{2 + |S| + |S| + 1 + O(\sqrt{|V|} \log{|V|}) + 2}{|V|^2} \geq 1 - \frac{1}{|V|}$, our \textsf{learn-separator} method returns a near-separator.

\end{proof}

\begin{theorem}\textbf{(Theorem~\ref{thm:DAG-complexity})}
Suppose $G=(V,E)$ is a rooted DAG with $|V| = n$, and maximum constant degree, $d$, with at most constant, $c$ directed (not necessarily disjoint) paths from root, $r$, to each vertex. Algorithm~\ref{alg:spanning-tree} learns a spanning tree of $G$ using $Q(n) \in O(n \log n)$ and $R(n) \in O(\log n)$ w.h.p.
\end{theorem}

\begin{proof}
By Lemma~\ref{lem:learn-separator-complexity}, each call to \textsf{learn-separator} method takes at most $O(c |V|) \in O(|V|)$ queries in $O(1)$ rounds, and using Lemma~\ref{lem:exist-separator} and Lemma~\ref{lem:learn-separator}, it returns a near-separator with probability at least $\frac{1}{d} \cdot \left( 1 - \frac{1}{|V|}\right)$. Then, it learns an edge through a call to \textsf{learn-parent} method using $O(|V|)$ queries in $O(1)$ rounds with probability at least $1 - 1/|V|$. Hence, for $|V| \geq 4d$, using a union bound, with probability at least $\frac{1}{2d}$, we have that our near-separator splits the DAG with vertices $V$ into two DAGs of size at least $\frac{|V|}{d+2}$. Hence, we have

\[ Q(n) = Q\left(\frac{n}{d+2}\right) + Q\left( \frac{n \cdot (d+1)}{d+2}\right) + O(n)\]
\[
R(n) = R\left(\frac{n}{d+2}\right) + R\left( \frac{n \cdot (d+1)}{d+2}\right) + O(1)
\]

Therefore, we have $Q(n) \in O(n \log n)$ and $R(n) \in O(\log n)$ in expectation. We can also use a Chernoff bound for sum of the independent geometric random variables (see~\cite{gt-adfai-02, DBLP:books/daglib/0012859}) to prove the bounds
with high probability as of the work of Afshar {\it et al.}~\cite[Theorem~13]{DBLP:conf/esa/AfsharGMO20}.

\end{proof}

\end{appendix}

\end{document}